\documentclass[twoside,twocolumn,compsoc,final,10pt,letterpaper]{IEEEtran}
\usepackage[T1]{fontenc}
\usepackage[latin9]{inputenc}
\usepackage{geometry}
\usepackage{framed}
\usepackage{amsbsy}
\usepackage{graphicx}
\usepackage{amssymb}

\makeatletter
\newtheorem{claim}{Claim}

\usepackage{amsmath}
\usepackage{psfrag}
\usepackage{pst-all}
\IEEEcompsoctitleabstractindextext{%
\begin{abstract}
Key management in wireless sensor networks faces several new challenges. The scale, resource limitations, and new threats such as node capture necessitate the use of an on-line key generation by the nodes themselves. However, the cost of such schemes is high since their secrecy is based on computational complexity. Recently, several research contributions justified that the wireless channel itself can be used to generate information-theoretic secure keys. By exchanging sampling messages during movement, a bit string can be derived that is only known to the involved entities. Yet, movement is not the only possibility to generate randomness. The channel response is also strongly dependent on the frequency of the transmitted signal. In our work, we introduce a protocol for key generation based on the frequency-selectivity of channel fading. The practical advantage of this approach is that we do not require node movement. Thus, the frequent case of a sensor network with static motes is supported. Furthermore, the error correction property of the protocol mitigates the effects of measurement errors and other temporal effects, giving rise to an agreement rate of over 97\,\%. We show the applicability of our protocol by implementing it on MICAz motes, and evaluate its robustness and secrecy through experiments and analysis.
\end{abstract}

}

\newcommand{\psfragLabel}[3][\footnotesize]{\psfrag{#2}[][]{#1{#3}}}
\newcommand{\psfragTick}[2][\tiny]{\psfrag{#2}[][][1][0]{#1{$#2$}}}
\newcommand{\psfragnTick}[2][\tiny]{\psfrag{-#2}[][][1][0]{#1{-$#2$}}}
\newcommand{\psfragLegend}[2]{\psfrag{#1}[][]{\footnotesize{#2}}}

\newcommand{\psfragFixChannels}[0]{%
\psfragTick{11}\psfragTick{12}\psfragTick{13}\psfragTick{14}\psfragTick{15}%
\psfragTick{16}\psfragTick{17}\psfragTick{18}\psfragTick{19}\psfragTick{20}%
\psfragTick{21}\psfragTick{22}\psfragTick{23}\psfragTick{24}\psfragTick{25}%
\psfragTick{26}%
}

\newcommand{\psfragFixRSSLevels}[0]{%
\psfragnTick{50}\psfragnTick{60}\psfragnTick{70}\psfragnTick{80}\psfragnTick{90}%
\psfragnTick{55}\psfragnTick{65}\psfragnTick{75}\psfragnTick{85}\psfragnTick{95}%
}

\newcommand{\psfragFixDiff}[1][\tiny]{%
\psfragnTick[#1]{5.5}\psfragnTick[#1]{5}\psfragnTick[#1]{4.5}\psfragnTick[#1]{4}%
\psfragnTick[#1]{3.5}\psfragnTick[#1]{3}\psfragnTick[#1]{2.5}\psfragnTick[#1]{2}%
\psfragnTick[#1]{1.5}\psfragnTick[#1]{1}\psfragnTick[#1]{0.5}\psfragTick[#1]{0}%
\psfragTick[#1]{0.5}\psfragTick[#1]{1}\psfragTick[#1]{1.5}\psfragTick[#1]{2}%
\psfragTick[#1]{2.5}\psfragTick[#1]{3}\psfragTick[#1]{3.5}\psfragTick[#1]{4}%
\psfragTick[#1]{4.5}\psfragTick[#1]{5}\psfragTick[#1]{5.5}\psfragTick[#1]{6}%
\psfragTick[#1]{8}\psfragTick[#1]{10}\psfragTick[#1]{12}\psfragTick[#1]{14}%
\psfragTick[#1]{16}\psfragTick[#1]{18}%
}

\newcommand{\psfragFixDensity}[0]{%
\psfragTick{0.0}\psfragTick{0.2}\psfragTick{0.4}%
\psfragTick{0.6}\psfragTick{0.8}\psfragTick{1.0}%
\psfragTick{1.2}%
}

\newcommand{\psfragFixTens}[1][\scriptsize]{%
\psfragTick[#1]{0}\psfragTick[#1]{10}\psfragTick[#1]{20}%
\psfragTick[#1]{30}\psfragTick[#1]{40}\psfragTick[#1]{50}%
\psfragTick[#1]{60}\psfragTick[#1]{70}\psfragTick[#1]{80}%
\psfragTick[#1]{90}\psfragTick[#1]{100}\psfragTick[#1]{110}%
\psfragTick[#1]{120}\psfragTick[#1]{130}\psfragTick[#1]{140}%
}

\newcommand{\psfragFixRSSBarplot}[0]{%
\psfragLabel{Wireless channels}{Wireless channels}%
\psfragLabel{Received signal strength [dBm]}{Received signal strength (dBm)}%
\psfragFixChannels{}%
\psfragFixRSSLevels{}%
}

\DeclareMathOperator{\dis}{dis}

\@ifundefined{showcaptionsetup}{}{%
 \PassOptionsToPackage{caption=false}{subfig}}
\usepackage{subfig}
\makeatother

\begin{document}

\title{Key Generation in Wireless Sensor Networks\\
Based on Frequency-selective Channels --\\
Design, Implementation, and Analysis}

\author{Matthias Wilhelm, Ivan Martinovic, and Jens B.~Schmitt %
\thanks{The authors gratefully acknowledge the financial support from the
Carl-Zeiss Foundation and the Forschungsschwerpunkt Ambient Systems
(amsys).%
}\medskip{}
\\
{\normalsize disco | Distributed Computer Systems Lab }\\
{\normalsize{} TU Kaiserslautern, Germany}\\
{\normalsize \{wilhelm,martinovic,jschmitt\}@cs.uni-kl.de}}

\maketitle

\section{Introduction}

\PARstart{S}{ecuring} wireless sensor networks (WSNs) has been one
of the main wireless network research areas in recent years. Especially
key generation and key management, which are at the heart of any security
design, pose new challenges because of the low computational capabilities
of wireless motes, their limited battery lifetime, and the broadcast
nature of wireless communication. Given these peculiarities, a large
number of key management protocols for WSNs has been proposed, often
fine-tuned between different performance vs.~security trade-offs
and adapted for specific WSN scenarios and their applications (for
a general overview see, e.g., \cite{WSNKeyDistSurvey,WSNkeyManageSurvey}).
However, most of these protocols follow a conventional cryptographic
approach, where the secret is based either on pre-distributed keys
or public-key schemes assuming more performance capable devices that
are able to generate and distribute the keys. Although there have
been efforts to adapt public key cryptographic protocols to the world
of WSNs (e.g., TinyECC~\cite{tinyecc}), these~adaptations usually
have a significant complexity and memory footprint as well as a high
energy consumption~(for energy analysis of public key schemes, see,
e.g., \cite{Energy-Analysis-PKC}). As an example, TinyECC (with optimizations)
requires roughly 20\,kB of ROM and 1.7\,kB of RAM \cite{tinyecc},
which is 15.6\,\% and 42.5\,\% of the overall available memory resources
of MICAz sensor motes, respectively, and single~operations require
computation time in the order~of~seconds.

Recently, there have been research contributions that follow an alternative
path towards key generation using an information-theoretic approach
to derive secrets from unauthenticated broadcast channels. Informally,
the general idea is similar to the quantum world, in which the laws
of quantum mechanics ensure that two spatially separated particles
experience highly correlated quantum states (called {}``quantum entanglement'').
Measuring the quantum properties of one particle discloses the knowledge
of another. However, in contrast to the mystical quantum nature, contributions
on key generation using wireless channel are concerned with conventional
physical signal propagation and, to some extent, its reciprocal behavior.
Specifically, recent results described by Mathur \emph{et al.}~\cite{Telepathy}
and Azimi-Sadjadi \emph{et al.}~\cite{Envelopes} justify that the
unpredictable multipath propagation and the resulting fading behavior
of wireless channel can be used to extract shared secret material.
Simply by exchanging messages that serve to sample the signal propagation
behavior, both transmitters can establish mutual secret information,
while an eavesdropper who also receives these messages still remains
completely ignorant of their channel measurements. Since the secrecy
of the extracted information is not based on computational complexity
as common to conventional public key cryptography, these protocols
are especially valuable to computationally limited wireless devices.
Yet, existing solutions require that the wireless devices move at
certain speeds to produce enough unpredictability in their signals.
Thus, the most prevalent applications of WSNs which are based on static
wireless motes make these protocols inapplicable. This brings us to
the contribution of this work, which abstains from this limitation
and provides a novel key generation protocol for static WSNs. The
main contributions of this paper are:
\begin{itemize}
\item Design of a robust key generation protocol with an error-correcting
property against channel deviations ($\rightarrow$ Section~\ref{sec:Protocol-Design}).
\item Implementation of the protocol on static MICAz sensor motes and analysis
of the protocol's robustness and the secrecy of derived keys, especially
with respect to dependencies between wireless channels ($\rightarrow$
Section~\ref{sec:Experimental-Analysis}).
\item Derivation of a stochastic model describing the secrecy of the protocol,
its validation using experimental data, and guidelines on increasing
the number of generated secret bits ($\rightarrow$ Section~\ref{sec:Analysis}).
\end{itemize}
In summary, we demonstrate the applicability of a key generation protocol
that takes advantage of the wireless channel behavior in static wireless
networks and analyze different trade-offs between it's robustness
to channel deviations and available secrecy.

\section{Related Work}

\begin{figure*}[t]
\hspace*{\fill}\subfloat[Alice's view when Bob transmits]{\hspace*{\fill}\psfragFixRSSBarplot{}\includegraphics[width=0.31\textwidth]{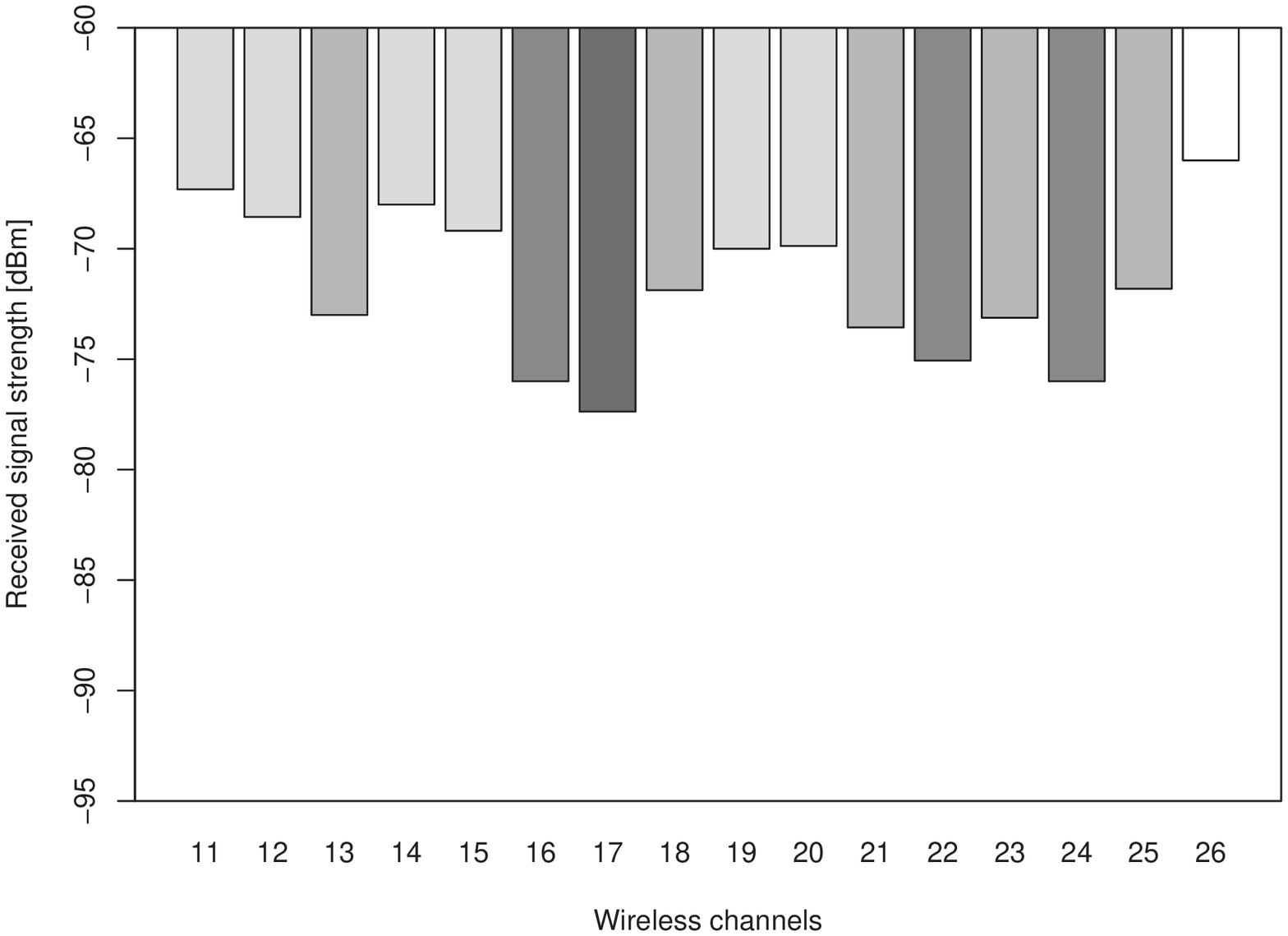}\hspace*{\fill}

}\hspace*{\fill}\subfloat[Bob's view when Alice transmits]{\hspace*{\fill}\psfragFixRSSBarplot{}\includegraphics[width=0.31\textwidth]{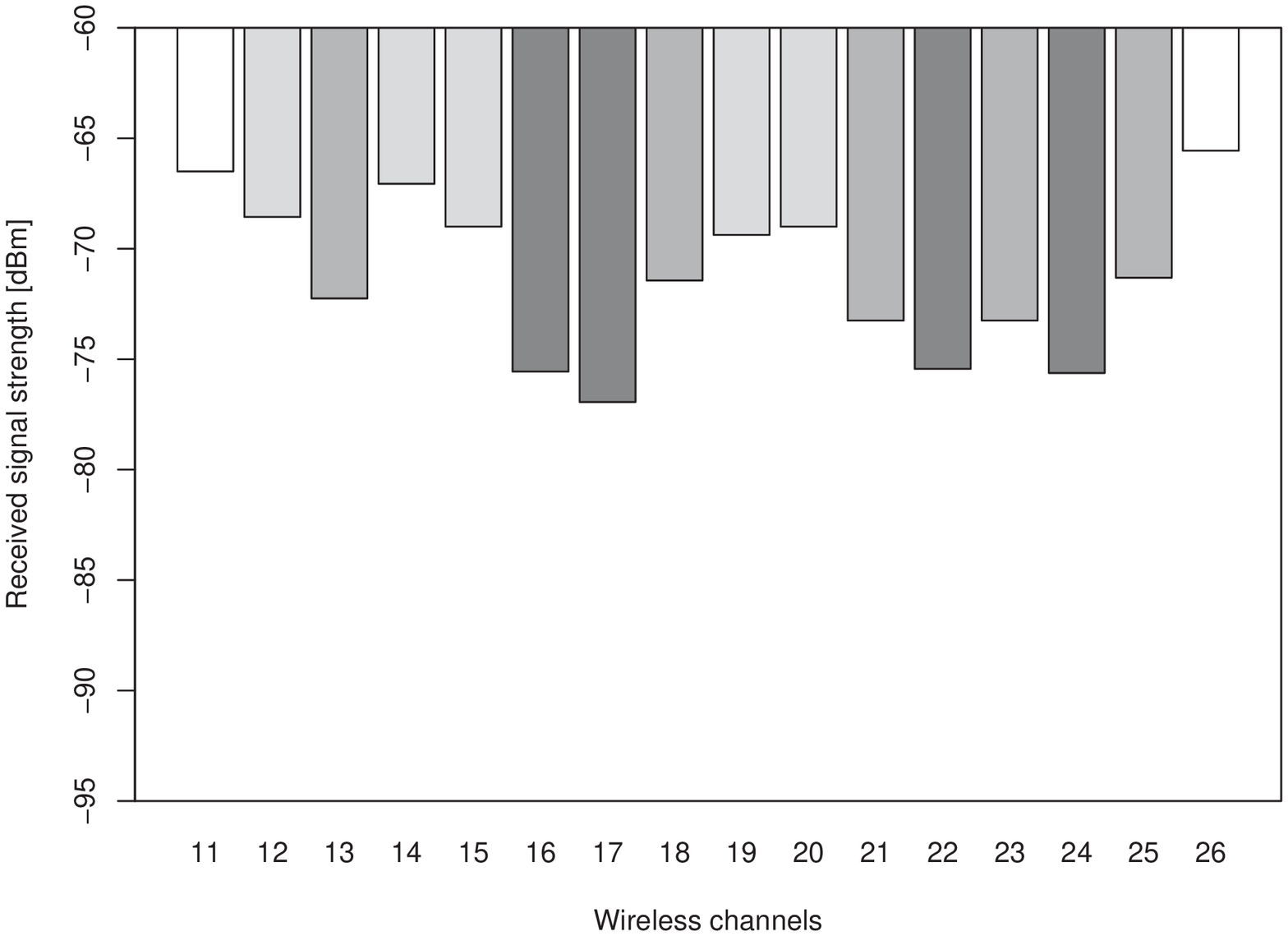}\hspace*{\fill}

}\hspace*{\fill}\subfloat[Deviations between views]{\hspace*{\fill}\psfragFixDiff{}
\psfragFixRSSBarplot{}
\psfrag{Difference between mean values [dB]}[][]{\footnotesize{Difference between means (dB)}}\includegraphics[width=0.31\textwidth]{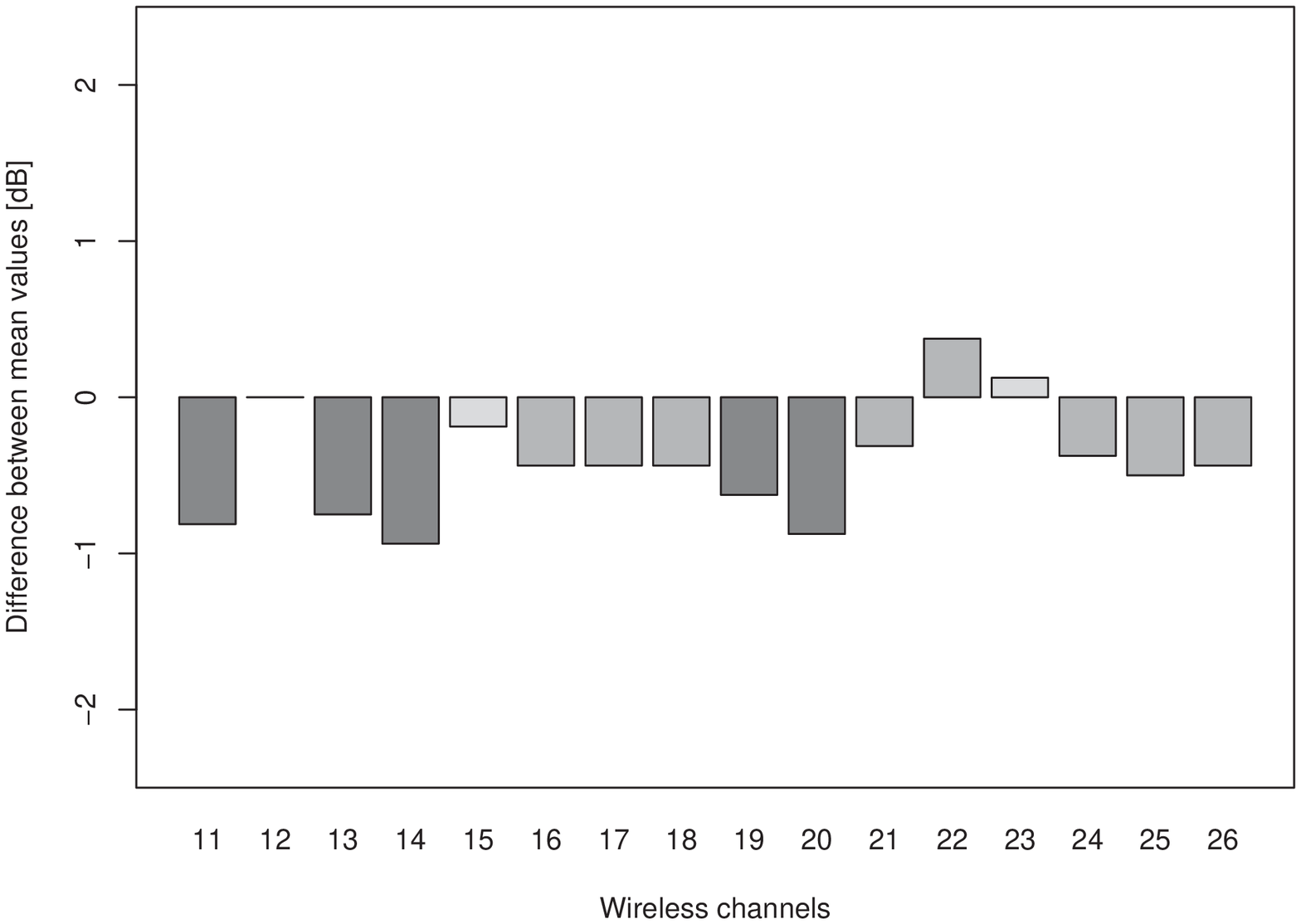}\hspace*{\fill}

}\hspace*{\fill}

\caption{The reciprocity of the wireless channel state is strong enough to
enable the extraction of shared secret information. \label{fig:Reciprocity}}

\end{figure*}

The use of physical properties was first considered in the context
of quantum cryptography. The laws of quantum mechanics ensure that
the same quantum states are observed by two spatially separated parties.
Several protocols have been proposed that exploit this property and
can guarantee the detection of eavesdroppers \cite{QuantumCoinTossing,QuantumTangledBell}.
The concept was generalized in the framework of information theory
by Maurer~\cite{MaurerKAwithCommonInfo}. Here, random sources observable
by three parties are considered: a source provides two strongly correlated
variables to two legitimate participants, and a weaker correlated
variable to an eavesdropper. This work shows that information-theoretically
secure keys can be derived from such sources even when an adversary
has partial access to the source of information. The theoretical concept
was instantiated for the use of wireless channels by the same research
group~\cite{MauWolf03,MaurerUnbreakableKeysfromNoise}. 

Several research contributions apply this concept to narrow-band communication
systems to generate secret keys from a wireless channel. Mathur \emph{et
al.}~\cite{Telepathy} use the randomness of the received signal
strength, which is introduced by movement, as a source for correlated
information, the so-called {}``radio-telepathy''. By frequent sampling
of the wireless channel both parties can create a sequence of channel
states that are strongly correlated because of the principle of reciprocity.
The fading behavior on a single sampling frequency is strongly dependent
on the physical position, and movement introduces uncertainty for
an adversary that is captured in these sequences. The degree of reciprocity
decreases rapidly in space, such that eavesdropping on sampling messages
does not allow to infer the channel state between the legitimate nodes.
The authors employ a level-crossing algorithm that uses two thresholds
for signal strength values to generate bit strings. For information
reconciliation, both parties detect mutual excursions by exchanging
suitable candidate regions in the sequence, thereby increasing the
chance to produce shared secret bits. The longer the required excursions
are, the more robust the scheme is against measurement errors. Yet,
in contrast to our work, their solution requires movement as a generator
of randomness and thus it is not applicable to static networks. Additionally,
the protocol introduced in \cite{Telepathy} requires more powerful
devices such as laptops or software-defined radios, as a high sampling
rate is necessary and a complex reconciliation mechanism is used to
avoid errors.

Azimi-Sadjadi \emph{et al.}~\cite{Envelopes} propose a similar protocol
that focuses mainly on the robustness of the key generation process,
i.e., tolerance against deviations in the wireless channel and a high
success rate. They employ a single threshold for detection of strong
deep fades introduced by movement, an event that is reliably detectable,
but also rare (in the order of Hz), again depending on the speed of
movement. Every sample is turned into an output bit of the protocol,
which leads to long sequences of {}``1''s, representing the absence
of deep fades, interrupted by shorter sequences of {}``0''s. The
resulting bit string is not directly usable as keying material, as
the uncertainty for an attacker is located in the \emph{position}
of the deep fades in the string. Thus, not all bits are equally unpredictable,
and the authors consider the use of randomness extractors to produce
uniformly random strings. No quantitative evaluation of secrecy is
given, but considering the use of deep fades only and the nature of
randomness extractors, we estimate that the use of this protocol results
in a lower secret bit rate than the approach in~\cite{Telepathy}.
Further results on such key extraction protocols, especially with
respect to the effectiveness in realistic scenarios, are given in
\cite{PatwariKeyExtraction}. 

Several other contributions use highly specialized hardware, such
as steerable antennas, ultra-wideband (UWB) radio or multi-antenna
systems with performance-capable processors \cite{EsparAntennaKGen,UWB-ChannelID,Multi-Antenna-KeyGen}.
In contrast, this paper focuses on the capabilities of conventional
{}``off-the-shelf'' sensor motes, without the need for additional
equipment.

In summary, it can be stated that current solutions provide valuable
insights into the feasibility of key generation using physical properties,
but several important issues still remain open. Especially the hardware
platform that benefits most from key generation schemes, wireless
sensor networks, is still unsupported. As current protocols require
movement and complex reconciliation to guarantee successful key generation,
the most prevalent\emph{ static scenarios} are not considered. Our
work closes this gap with a protocol that can be used even on the
most resource-constrained hardware and is specially designed for static
environments. Our initial results are presented in \cite{WMS09-1}
and \cite{WMS10-1}. This work extends our previous results and finalizes
them. It offers extensive experimental analysis using IEEE~802.15.4
technology, an in-depth evaluation of secrecy, especially with respect
to dependencies between wireless channels, and a stochastic model
that captures the behavior of the proposed protocol and provides predictions
on the different trade-offs between security and robustness.

\section{Concept \label{sec:Concept}}

In this section, we introduce the concept of key generation using
the frequency-selectivity of wireless channels. As we base the secrecy
of our protocol on our ability to extract secrets at two different
locations, we require two things from the wireless channel: strongly
correlated information between the two parties and high uncertainty
about the generated keying material for adversaries.

\subsection{Robustness Considerations}

The principle of channel reciprocity states that two receivers experience
the same properties of the wireless channel if their role as sender
and receiver is exchanged, given that the time interval is shorter
than the coherence time~$t_{c}$ of the channel. As we mainly consider
static scenarios, the reciprocity between nodes is strong, even if
the sampling rate is small owning to the limited capabilities of the
considered hardware. Measurements show that we are able to distinguish
signal strengths even when using fine-grained levels. As an example
of this behavior, Fig.~\ref{fig:Reciprocity} presents such measurements
from a single constellation of sender and receiver. On each channel,
16 sampling messages are exchanged to generate robust results. The
experiments exhibit bounded deviations, the RSS indicator reported
by the hardware is able to capture the channel state accurately enough
to enable successful key generations.

Imperfect reciprocity directly influences the robustness of the proposed
key generation protocol, as deviations in the view on the channel
lead to disagreement in the produced bit strings. A second factor
is measurement errors caused by noise, both in the measurement circuits
and the wireless channel. All of these deviations must be corrected
to successfully generate secret keys. Our experimental analysis presented
in Section~\ref{sub:Imp-Protocol-Robustness} will show that these
deviations are sufficiently small for different indoor scenarios,
and secrets can be generated reliably even on stock sensor motes.

\subsection{Security Considerations}

\begin{figure}[t]
\hspace*{\fill}\input{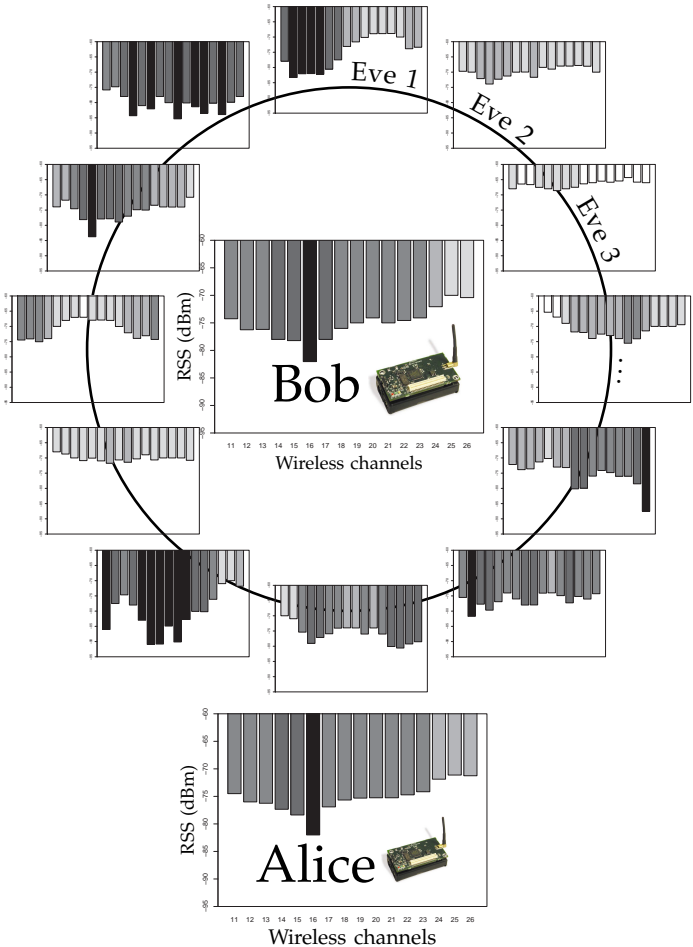}\hspace*{\fill}

\caption{Taking advantage of spatial and frequency selectivity of multipath
fading experienced in wireless channels. Even if Eve takes positions
on a circle with 10\,cm radius around the position of the legitimate
transmitter (Bob), the measured signalprints are significantly different
from Bob's measurements. \label{fig:10cm}}

\end{figure}

The unpredictability of the channel state is the most important aspect
when considering the wireless channel as a source of randomness, as
it directly affects the available secrecy. In the related work \cite{Telepathy,Envelopes},
the \emph{spatial} selectivity of the wireless channel due to movement
is used to generate secret bits. In this work, we show that the frequency-selectivity
of multipath fading is a viable alternative to generate secret information
using the wireless channel, without requiring node movement.

In general, wireless signals are not traveling on a single path from
a sender to a receiver, but arrive from several directions at the
receiver, i.e., the signal exhibits multipath propagation characteristics.
Each path is affected by different attenuations and phase shifts,
and the resulting signal at the receiver is a combination of all signal
paths by wave interference, resulting in a channel response depending
on many variables. A small variation in phase, e.g., by using a different
carrier frequency, leads to unpredictable changes in the signal strength,
even when signal paths are unchanged. This behavior is captured by
the impulse response of the wireless channel, consisting of a number
of time-shifted Dirac pulses $\delta$, and considering $L$ signal
paths \[
\mathrm{h}(\tau)=\sum_{l=1}^{L}\alpha_{l}\mathrm{e}^{j\phi_{l}}\delta\left(\tau-\tau_{l}\right),\]
with different values of each path for the amplitude $\alpha_{l}$,
phase shift $\phi_{l}$ and delay $\tau_{l}$, acting as random variables.
Because of phase shifts, interference effects can lead to signal cancellation
or amplification, depending on the relative phase shifts.

To show the magnitude of these effects, we conducted an experiment
to evaluate the selectivity of the channel both with respect to position
and carrier frequency. Fig.~\ref{fig:10cm} shows the uncertainty
of an adversary even if he is positioned very close Bob. Each barplot
represents the received signal strength measurements on 16 channels
in the 2.4\,GHz range available on the MICAz platform. The sensor
mote acting as Alice was placed in a fixed position on a desk, Bob
was placed in an adjacent room, such that both were separated by a
wall, and the channel response was sampled from 12 positions on a
10\,cm radius around Bob's initial position. The results show that
the multipath effects are strong, and even if an attacker has knowledge
of the environment and the positions of Alice and Bob; the channel
behavior is unpredictable. Even ray-tracing approaches are unable
to capture this behavior precisely, as a highly accurate model of
the environment capturing minimal phase shifts would be required.
Extensive results on the amount of uncertainty for an adversary obtained
in our experiments are given in Section~\ref{sub:Imp-Secrecy}.

\subsection{System Model}

We are interested in the \emph{amount} of uncertainty that an adversary
experiences. Information theory introduces the notion of (Shannon)
entropy to quantify the average amount of information of a discrete
random variable, making it suitable for capturing the amount of uncertainty
an attacker experiences. In this section, we derive a stochastic model
of the system enabling us to evaluate the secrecy of the proposed
protocol based on signal strength distributions of real-world measurements.

\subsubsection{Secrets from the Wireless Channel}

The state of the wireless channel for a specified frequency at a certain
point in time is captured by the \emph{discrete} random variable $C$,
that is, we assume that only finite precision can be achieved in channel
state acquisition. Possible sources for this variable are, for example,
the complex impulse response of the channel, or as in our case, the
received signal strength. The outcome of $C$ is stable during channel
coherence time, which depends on the speed of movement. In static
scenarios on which we focus this time is very long, enabling us to
take several samples and use mean values as outcomes of $C$. 

Both Alice and Bob have access to the wireless channel and can exchange
sampling messages. Each can monitor one of the random variables \begin{eqnarray*}
X_{\mathrm{Alice}} & = & C_{\mathrm{Alice}}+N_{\mathrm{Alice}},\\
X_{\mathrm{Bob}} & = & C_{\mathrm{Bob}}+N_{\mathrm{Bob}},\end{eqnarray*}
with $C_{x}$ being the measured channel state at the respective position
and $N_{x}$ being random variables representing the noise processes
that introduce errors in the channel state estimations. With the help
of channel reciprocity we can assume that $C_{\mathrm{Alice}}=C_{\mathrm{Bob}}=C$,
i.e., both parties experience the same channel properties in their
exchanged sampling messages. The mutual information that the channel
provides is described by \[
\mathrm{I}\left(X_{\mathrm{Alice}},X_{\mathrm{Bob}}\right)=\mathrm{H}\left(X_{\mathrm{Alice}}\right)-\mathrm{H}\left(X_{\mathrm{Alice}}\middle|X_{\mathrm{Bob}}\right)\leq\mathrm{H}\left(C\right).\]
The conditional entropy $\mathrm{H}\left(X_{\mathrm{Alice}}\middle|X_{\mathrm{Bob}}\right)$
is zero if the channel is noiseless, and then the amount of shared
information which Alice and Bob gain from monitoring the wireless
channel is quantified by the entropy $\mathrm{H}\left(C\right)$ of
the channel state variable, given by \[
\mathrm{H}\left(C\right)=-\sum_{c\in\mathcal{C}}\mathrm{p}\left(c\right)\log\mathrm{p}\left(c\right),\]
where $\mathrm{p}\left(c\right)$ denotes the probability mass function
of $C$ and $\mathcal{C}$ its support. This also represents the maximum
attainable mutual information from the wireless channel, because the
noise term $N=N_{\mathrm{Alice}}-N_{\mathrm{Bob}}$ that captures
deviations in the measurements has a negative effect on the mutual
information \cite{MaurerUnbreakableKeysfromNoise}. We propose a reconciliation
mechanism to correct the errors introduced at this point, which is
presented in the next section. An experimental evaluation of the magnitude
of measurement errors and the effects on secrecy is given in Section~\ref{sec:Experimental-Analysis},
as we aim to quantify the amount of secrecy using the propagation
properties of realistic wireless channels.

An eavesdropper who can also monitor the sampling message to infer
the channel state $C$ between Alice and Bob measures $X_{\mathrm{Eve}}=C_{\mathrm{Eve}}+N_{\mathrm{Eve}}$.
As $C$ and $C_{\mathrm{Eve}}$ de-correlate rapidly in space, as
shown empirically by Mathur \emph{et al.}~in \cite{Telepathy}, the
mutual information $\mathrm{I}\left(X_{\mathrm{Alice}},X_{\mathrm{Eve}}\right)$
and $\mathrm{I}\left(X_{\mathrm{Bob}},X_{\mathrm{Eve}}\right)$ are
approaching zero if the distance is greater than a wavelength, thus
eavesdropping on the sampling messages does not help Eve to infer
information on $C$. The entropy $\mathrm{H}\left(C\right)$ stands
against Eve, it quantifies the amount of uncertainty in the channel
state for Eve accurately. 

However, the information on a single channel is limited, and a way
must be identified to increase the amount of shared information between
Alice and Bob. Two possibilities of increasing entropy can be considered:
\emph{(i)} create a random process $C\left(t\right)$ by moving the
devices (reducing the channel coherence time), or \emph{(ii)} probe
multiple channels to exploit the frequency-selectivity of the wireless
channel. The first approach is followed in \cite{Telepathy,Envelopes},
which is effective and easy to analyze for its secrecy but, as pointed
out, poses several problems for an adoption in WSNs. To support static
networks, we propose and evaluate the second approach in this work.

\subsubsection{Multiple Channels}

We now consider the random vector $\mathbf{C}=\left(C_{1},\ldots,C_{n}\right)$,
measured on $n$ different frequencies (channels). In this case, Alice
measures $\mathbf{X}_{\mathrm{Alice}}=\left(X_{\mathrm{Alice}}^{(1)},\ldots,X_{\mathrm{Alice}}^{(n)}\right)$
and Bob measures the corresponding vector $\mathbf{X}_{\mathrm{Bob}}$,
which both can be used to obtain the mutual information \[
\mathrm{I}\left(\mathbf{X}_{\mathrm{Alice}},\mathbf{X}_{\mathrm{Bob}}\right)=\mathrm{H}\left(\mathbf{X}_{\mathrm{Alice}}\right)-\mathrm{H}\left(\mathbf{X}_{\mathrm{Alice}}\middle|\mathbf{X}_{\mathrm{Bob}}\right)\leq\mathrm{H}\left(\mathbf{C}\right),\]
assuming reciprocity on all channels and $\mathrm{H}\left(\mathbf{C}\right)$
being the \emph{joint} entropy over all channels, given by \[
\mathrm{H}\left(\mathbf{C}\right)=-\sum_{c_{j}\in\mathcal{C}_{i}}\mathrm{p}\left(c_{1},\ldots c_{n}\right)\log\mathrm{p}\left(c_{1},\ldots,c_{n}\right).\]
If the elements in the random vector are independent, then the amount
of uncertainty can directly be evaluated using the entropy values
from individual channels, $\mathrm{H}\left(\mathbf{C}\right)=\sum_{i=1}^{n}\mathrm{H}\left(C_{i}\right)$.
This value represents an upper bound on the joint entropy, as known
dependencies between the variables enable predictions and reduce the
overall uncertainty of Eve. Wireless channels experience correlated
fading if the distance between the center frequencies is smaller than
the coherence bandwidth. This is the case for our hardware platform,
MICAz sensor motes. If it were not the case, the secrecy analysis
would be fairly easy. Yet, in Section~\ref{sec:Experimental-Analysis}
we will not make that simplifying assumption. Therefore, we analyze
the dependency structure to evaluate the amount of uncertainty, i.e.,
the secrecy of keys generated by the presented protocol. We do this
with respect to the following adversarial model.

\subsubsection{Adversarial Model}

One important aspect for the quantification of secrecy of such a scheme
is to define the abilities of an adversary, in the same~way as it
is necessary when evaluating cryptographic security protocols. For
instance, computationally~unbounded~attackers can break Diffie-Hellman
key agreements with ease because they can solve any problem that relies
on computational complexity. Similarly, an attacker who can take \emph{exactly}
the same physical~positions as legitimate sensor~nodes can break
our key generation~protocol. Yet, with realistic constraints on an
attacker, the security of the protocol can be analyzed quantitatively.

An adversary has several options to attack the secrecy of the key
generation protocol. It can eavesdrop on the wireless~channel and
observe both the content of the messages and the signal strengths
that it can experience at its position. As the content of the messages
carries no~information and the signal strength de-correlates rapidly
in space, this gives it very little information on the channel~state
between Alice and Bob. Thus, eavesdropping is not an effective option.
With its presence, it can only prevent Alice and Bob from exchanging
secret information in plaintext over the wireless~channel. 

The best~attack~vector is to model the multipath~channel between
Alice and Bob, taking into consideration the hardware and environment,
and then infer the signal strength values. Knowledge to aid an attacker
in this modeling can come from plans of the building for indoor scenarios
or from observations of the environment, from the positions of Alice
and Bob by observation of the sensor motes, or via positioning~methods
using wireless signals, such as triangulation. While the effects of
path loss and shadowing on the line-of-sight (LOS) connection between
the two nodes are predictable (e.g., using ray-tracing methods~\cite{WISEraytracing}),
the resolution of the multipath components is very challenging. To
refine its model, an adversary is allowed to do measurements with
similar hardware off-site. The only assumption here is that the attacker
cannot measure at the very same positions of the legitimate sensors
during operation, because this is equivalent to a node capture which
discloses the key directly.

Given this information, we can model the knowledge of an adversary
by limiting possible signal strengths to the distribution of signal
strengths of similar positions. This can be achieved by using the
distribution of signal strength values from channel propagation models,
that is, he can generate accurate distributions for $C_{i}$ between
Alice and Bob. This allows quantifying the amount of uncertainty that
the attacker experiences; we can quantify its expected uncertainty
with the entropy $\mathrm{H}\left(\mathbf{C}\right)$ of the signal
strength distributions of the wireless channel.

\section{Protocol Design \label{sec:Protocol-Design} }

\begin{figure*}
\hspace*{\fill}\psset{xunit=.95\textwidth,yunit=7.7cm,runit=1cm}
\begin{pspicture}(0,-.05)(1,1)


\rput[bl](.375,.95){\large{Alice}}
\rput[bl](.575,.95){\large{Bob}}

\rput[bl](0,.925){Sampling Phase}
\rput[bl](0,.55){Key Generation Phase}
\rput[bl](0,.195){Key Verification Phase}

\psset{linestyle=dashed}
\psframe(0,.65)(1,.925)
\psframe(0,.275)(1,.55)
\psframe(0,0)(1,.195)
\psset{linestyle=solid}

\psset{linewidth=2pt}
\psline(.4,.94)(.4,.65)
\psline(.6,.94)(.6,.65)
\psline(.4,.55)(.4,.275)
\psline(.6,.55)(.6,.275)
\psline(.4,.195)(.4,0)
\psline(.6,.195)(.6,0)
\psset{linewidth=1pt}

\psline[]{->}(.4,.87)(.6,.87)
\rput[bc](.5,.895){switchChannel()}
\psline[]{->}(.4,.8)(.6,.8)
\rput[bc](.5,.825){sampleChannel()}
\psline[]{<-}(.4,.725)(.6,.725)
\rput[bc](.5,.75){sampleChannel()}
\psline[]{->}(.4,.4)(.6,.4)
\rput[bc](.5,.425){$\left(\mathbf{T},\mathbf{P}\right)$}
\psline[]{<-}(.4,.135)(.6,.135)
\rput[bc](.5,.16){$\textrm{h}(secret^\prime)$}
\psline[]{->}(.4,.025)(.6,.025)
\rput[bc](.5,.05){success()}

\psline[]{->}(.4,.64)(.4,.56)
\rput[br](.39,.575){$\boldsymbol{\mu}$}
\psline[]{->}(.6,.64)(.6,.56)
\rput[bl](.61,.575){$\boldsymbol{\mu}^{\prime}$}

\psline[]{->}(.4,.27)(.4,.2)
\rput[br](.39,.225){$secret$}
\psline[]{->}(.6,.27)(.6,.2)
\rput[bl](.61,.225){$secret^{\prime}$}

\rput[br](.39,.725){$m_i^{(j)}=\mathrm{RSS}(sample)$}
\rput[br](.39,.65){$\mu_i=\frac{1}{k} \sum_{j=1}^k m_i^{(j)}$}
\rput[bl](.61,.775){$m_i^{\prime (j)}=\mathrm{RSS}(sample)$}
\rput[bl](.61,.65){$\mu_i^{\prime}=\frac{1}{k} \sum_{j=1}^k m_i^{\prime (j)}$}
\psline[]{<->}(.025,.675)(.025,.9)
\psline[]{<->}(.975,.675)(.975,.9)
\rput[bc]{90}(.015,.7875){\scriptsize{$i=1,...,n$}}

\psline[]{<->}(.075,.85)(.075,.72)
\psline[]{<->}(.925,.85)(.925,.72)
\rput[bl]{90}(.07,.705){\scriptsize{$j=1,...,k$}}
\psline[linewidth=.5pt,linestyle=dotted](.075,.72)(.925,.72)
\psline[linewidth=.5pt,linestyle=dotted](.075,.85)(.925,.85)

\rput[bl](.025,.48){$t_i=\mathrm{chooseTolerance}\left(\mathbf{m}_i,errors\right)$}
\rput[bl](.025,.43){$q_i=\mathrm{q}_{t_i}(\mu_i)$}
\rput[bl](.19,.43){$P_i=q_i-\mu_i$}
\rput[bl](.025,.35){$\mathbf{P}=\left(P_1,\ldots,P_n\right)\quad\mathbf{T}=\left(t_1,\ldots,t_n\right)$}
\rput[br](.39,.28){$secret=\mathrm{bin}(q_1)\mid\mid\cdots\mid\mid\mathrm{bin}(q_n)$}
\rput[bl](.65,.47){$t_i=\mathrm{getTolerance}\left(\mathbf{T}\right)$}
\rput[bl](.65,.41){$P_i=\mathrm{getReconcileToken}\left(\mathbf{P}\right)$}
\rput[bl](.65,.35){$q_i^\prime =\mathrm{q}_{t_i}\left(\mu_i^\prime +P_i\right)$}
\rput[bl](.61,.28){$secret^\prime =\mathrm{bin}(q_1^\prime )\mid\mid\cdots\mid\mid\mathrm{bin}(q_n^\prime)$}
\psline[]{<->}(.015,.425)(.015,.53)
\rput[bl]{90}(.012,.42){\tiny{$i=1..k$}}
\psline[linewidth=.5pt,linestyle=dotted](.015,.423)(.4,.423)
\psline[linewidth=.5pt,linestyle=dotted](.015,.53)(.4,.53)

\psline[]{<->}(.95,.35)(.95,.525)
\rput[tl]{90}(.955,.35){\scriptsize{$i=1,...,k$}}
\psline[linewidth=.5pt,linestyle=dotted](.6,.35)(.95,.35)
\psline[linewidth=.5pt,linestyle=dotted](.6,.525)(.95,.525)

\rput[br](.39,.1215){$\mathrm{h}\left(secret\right)==\mathrm{h}\left(secret^\prime\right)?$}
\rput[br](.39,.07){False: Choose new tolerances $t_i$}
\rput[br](.39,.01){True: Key verified.}
\psline[](.05,.1)(-.01,.1)
\psline[](-.01,.1)(-.01,.525)
\psline[]{->}(-.01,.525)(-.002,.525)
\rput[bc]{90}(-.025,.3){$errors$\,++}

\end{pspicture}\hspace*{\fill}

\caption{Key generation protocol. The protocol operates in three phases; \emph{(i)}
the acquisition of channel state estimates, \emph{(ii)} error correction
using multi-level quantization and \emph{(iii)} secret verification.
The channel state estimates can be reused if the chosen tolerance
values are too small for the experienced deviations. \label{fig:Protocol}}

\end{figure*}
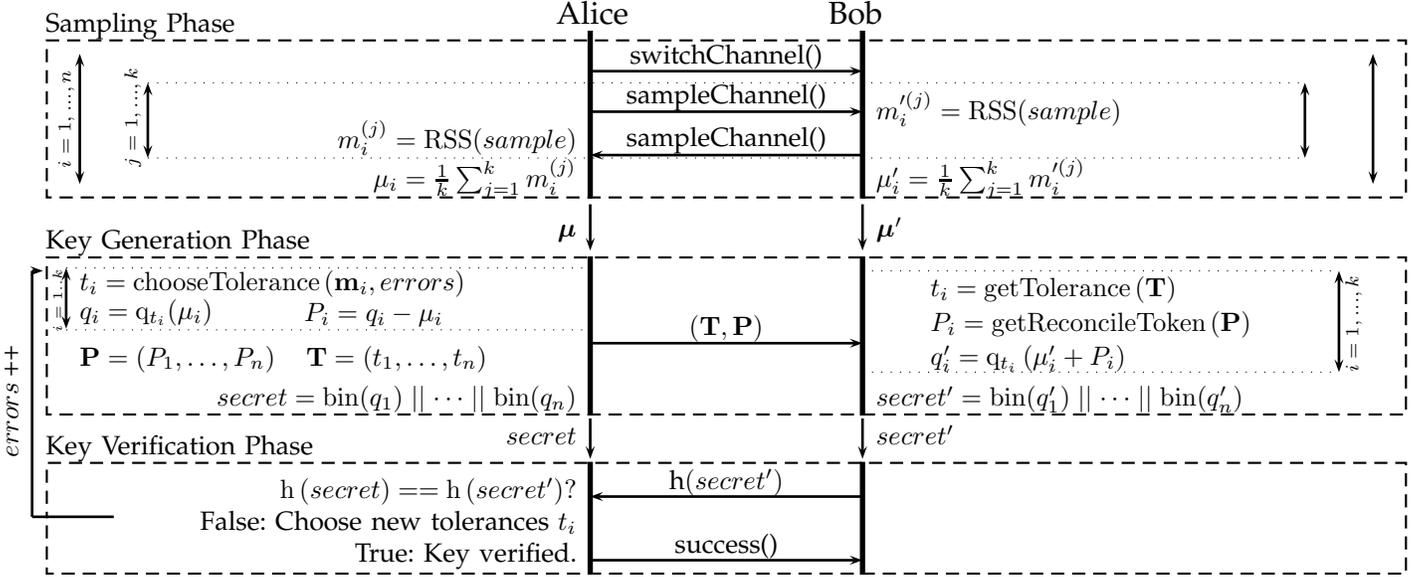

In this section, we present our novel key generation protocol suitable
even for limited hardware capabilities by using a performance-aware
design, specifically with WSNs in mind.

In the following, we conduct measurements by sampling RSS~values
on a set of $n$ different frequencies~$\mathcal{F}=\left\{ f_{1},\ldots,f_{n}\right\} $
(also referred to as \emph{channels}). The number of samples taken
is $k$, i.e., for each channel $f_{i}$ we collect a set of measurements
$\mathbf{m}_{i}=\left\{ m_{i}^{(1)},\ldots,m_{i}^{(k)}\right\} $.
To increase the error tolerance of our scheme, we calculate the mean
value $\mu_{i}=\frac{1}{k}\sum_{j=1}^{k}m_{i}^{(j)}$ of these RSS
samples. We view this mean value as the random variable~$C_{i}$,
which is distributed depending on the characteristics of wireless
propagation, e.g., following the commonly assumed Rayleigh or Ricean
distributions. The means of all $n$~channels are combined to the
random vector $\mathbf{C}=\left(C_{1},\ldots,C_{n}\right)$. A realization,
the outcome of our measurements is $\boldsymbol{\mu}=\left(\mu_{1},\ldots,\mu_{n}\right)$,
with $\mu_{i}\in\mathcal{M}=\left[\mu_{\mathrm{min}},\mu_{\mathrm{max}}\right]$,
the range of signal strength values that can be measured by the hardware
platform. We assume that only a finite precision in the measurements
can be achieved. As an example, in our wireless sensor network testbeds
we used $\mathcal{M}=\left[\textrm{-}104,\textrm{-}40\right]$\,dBm,
with a precision depending on the number of samples taken, since each
RSS~sample is integer~valued. We associate $\mathcal{M}$ with the
distance~function $\dis:\mathcal{M}\times\mathcal{M}\rightarrow\mathbb{R}^{+}$
defined as $\dis\left(\mu_{i},\mu_{i}^{\prime}\right):=\left|\mu_{i}-\mu_{i}^{\prime}\right|$,
which is the difference in~dB in our case. Thus, $\mathcal{M}$ together
with this distance~function constitutes a \emph{metric~space, }a
necessary prerequisite for the discussion of our error correction
scheme.

\subsection{Multi-level Quantization}

To successfully repair deviations in channel state measurements between
Alice and Bob, we use multi-level quantization to make close measurements
equal. In general, our quantization scheme~$\mathcal{Q}$ uses a
subset of the metric~space~$\mathcal{M}$, $\mathcal{Q}=\left\{ q_{1},\ldots,q_{K}\right\} \subseteq\mathcal{M}$,
with a total of $K$~elements, the \emph{quantization levels}. The
most important~property of the quantization scheme is the \emph{tolerance~}$t$
of the quantization~$\mathcal{Q}$. This is the largest~distance
for which an $m\in\mathcal{M}$ is mapped~uniquely, i.e., for all
$\mu_{i}\in\mathcal{M}$, we have $\dis\left(\mu_{i},q\right)<t$
for at most one~$q\in\mathcal{Q}$. Therefore, all values~$\mu_{i},\mu_{i}^{\prime}$
are mapped to~$q$ given their distance to $q$ is small~enough.

\subsubsection{Construction}

We choose $K$~elements~of~$\mathcal{M}$ with the same~distance~$d$
between quantization levels, where $p=\left\lceil \log_{2}K\right\rceil $
is the number~of~bits that are needed to identify a level. This
equidistance ensures that the tolerance $t$ is the same for all values
in $\mathcal{M}$. We denote this quantization as $\mathcal{Q}_{t}=\left\{ q_{1},\ldots,q_{K}\right\} $,
the bijective mapping to the binary~representation as $\mathrm{bin}:\mathcal{Q}_{t}\rightarrow\left\{ 0,1\right\} ^{p}$,
which maps quantized values to binary~strings. Since $\mu_{\mathrm{min}}$~and~$\mu_{\mathrm{max}}$
are both fixed~values, the distance~$d$ between neighboring quantization
levels is reduced as the number~of~levels increases. The relation
is given by $d=\frac{\left|\mu_{\mathrm{max}}-\mu_{\mathrm{min}}\right|}{K}$.
The tolerance of this scheme is given by~$t=\frac{d}{2}$, since
all~levels are evenly~spaced. The number of levels therefore directly
affects the tolerance of the quantization scheme, therefore, when
fewer levels are considered, larger deviations can be repaired. The
process~of~quantization maps the value~$\mu$ to the levels~$q$
with a minimal~distance in~$\mathbb{R}$, formally \[
\textrm{q}_{t}\left(\mu\right)=\arg\min_{q\in\mathcal{Q}_{t}}\dis\left(\mu,q\right).\]
For example, consider the quantization scheme \[
\mathcal{Q}_{1}=\left\{ \textrm{-}104,\textrm{-}102,\ldots,\textrm{-}42,\textrm{-}40\right\} \]
 with $32$~levels and tolerance~$t=1$ for our metric~space~$\mathcal{M}$.
For this, the measured value~$\mu=\textrm{-}71.424$\,dBm is quantized
to the level~$q=\textrm{-}72$. This ensures that values with distances
smaller than $1$\,dB are mapped to equal levels.

\subsubsection{Tolerance Properties of the Quantization Scheme}

The amount of uncertainty is reduced in this process as several values
are mapped to the same quantization level, but at the same time the
tolerance for deviations is increased. Thus, we can trade between
robustness and secrecy by choosing a $\mathcal{Q}_{t}$ with a suitable
tolerance~$t\in\mathbb{R}$ that is able to correct~errors in measurements
given $\dis\left(\mu,\mu^{\prime}\right)<t$. 

Still, some constellations are possible, such that $\mu$~and~$\mu^{\prime}$
are mapped to two different~levels (e.g., given $\mathcal{Q}_{1}$,
$\mu=\textrm{-}70.9$\,dBm and $\mu^{\prime}=\textrm{-}71.1$\,dBm
are mapped to $\textrm{-}70$ and $\textrm{-}72$, respectively).
To correct these error~patterns, we need to send a public piece of
information $P$ that helps Bob to reconcile his measurement and recover
the same~quantization results as Alice. Of course, at the same time
$P$ should reveal no new information to Eve. 

Our construction is straightforward: Alice calculates $P=\textrm{q}_{t}\left(\mu\right)-\mu$,
the shift that is necessary from $\mu$ to the corresponding quantization
value $q=\textrm{q}_{t}\left(\mu\right)$, and uses $q$ as her secret
information. This shift is always smaller than or equal to~$t$,
and therefore reveals only information that is discarded by Alice
and Bob anyway due to the quantization property. Alice then sends
$P$ via public~channel to Bob, who uses $P$ to generate the same
level~$q$ using his measurement~$\mu^{\prime}$ by calculating
$q=\textrm{q}_{t}\left(\mu^{\prime}+P\right)$. 
\begin{claim}
By using this reconciliation scheme, both Alice and Bob obtain $q$,
given $\dis\left(\mu,\mu^{\prime}\right)<t$.\end{claim}
\begin{proof}
Considering $\dis\left(\mu,\mu^{\prime}\right)<t$, then the distance
between the mean values is unchanged when both sides are shifted by
$P$, i.e., $\dis\left(\mu+P,\mu^{\prime}+P\right)<t$. From the construction
of $P$, we can infer that $q=\textrm{q}_{t}\left(\mu\right)=\mu+P$,
and thus $\dis\left(q,\mu^{\prime}+P\right)<t$. Finally, as the quantization
distance of the used scheme is $t$, $\mu^{\prime}+P$ is uniquely
mapped to $q$ by Bob as well, $\textrm{q}_{t}\left(\mu^{\prime}+P\right)=q$. 
\end{proof}

\subsection{Key Generation Protocol}

The proposed key generation protocol operates in three phases. In
the \emph{sampling phase}, the channel state is acquired, and due
to the reciprocity of the wireless channel state~information strongly
correlated measurements are collected by the two legitimate parties
in the protocol. In the \emph{key generation phase}, these deviations
are corrected, resulting in a secret bit string that is guaranteed
to be equal if the experienced deviations are bounded and suitable
quantization levels are chosen. The \emph{key verification phase}
ensures correct key agreement. The complete~protocol is shown in
Fig.~\ref{fig:Protocol}. We used a straightforward protocol for
the ease of presentation of the protocol analysis, but we also experimented
with several protocol optimizations that can further increase the
robustness and secrecy of the protocol, as presented in Section~\ref{sub:Protocol-Optimizations}.

\subsubsection{Sampling Phase}

In this initial phase, Alice and Bob exchange sampling messages over
the set of available wireless channels. For each of the $n$ frequencies
in $\mathcal{F}$, Alice and Bob exchange $k$ messages and each one
stores a set of measured RSS values $m_{i}$ or $m_{i}^{\prime}$,
respectively. Alice initiates the message exchanges, Bob answers incoming
sampling messages as fast as possible for a maximum of channel reciprocity.
Due to constraints of the mote hardware, the samples must be collected
in an interleaved manner, such that the state of the wireless channel
can change slightly, contributing to the noise terms $N_{\mathrm{Alice}}$
and $N_{\mathrm{Bob}}$. However, by using several sampling messages
per channel, the adverse effects of such short term deviations can
be mitigated. The mean values $\mu_{i}=\frac{1}{k}\sum_{j=1}^{k}m_{i}^{(j)}$
are then computed by Alice, while Bob proceeds similarly with $\mu_{i}^{\prime}$.
Thus, after finishing the sampling phase, both Alice and Bob possess
the vectors of channel state information $\boldsymbol{\mu}$ and $\boldsymbol{\mu}^{\prime}$
that capture the fading behavior of the wireless channel.

\subsubsection{Key Generation Phase}

The gathered mean value vectors $\boldsymbol{\mu}$ and $\boldsymbol{\mu}^{\prime}$
contain secret information that can be used as secret keys, but after
the sampling phase these vectors are unlikely to agree. The \emph{key
generation phase} uses information reconciliation based on the introduced
error correction scheme to produce a bit string that is equal on both
sides, without discarding shared bits or revealing information to
eavesdroppers. Alice chooses a set of tolerance values $\mathbf{T}=\left(t_{1},\ldots,t_{n}\right)$
based on the variance of its RSS values $\boldsymbol{m}_{i}$ and
the number of experienced verification errors from potential previous
runs. We used the same starting tolerance value $t_{i}=1$ for all
channels in our experiments and analysis, which achieves a high rate
of successful key agreements as well as good secrecy, as shown experimentally
with our implementation. However, the choice of tolerance values strongly
influences the robustness and secrecy trade-off, and considering optimization
at this point is useful (see a corresponding discussion in Section~\ref{sub:Protocol-Optimizations}). 

Alice uses the tolerance values to instantiate the appropriate quantization
functions $\textrm{q}_{t_{i}}$ and applies them on her mean values
$\mu_{i}$ to generate the values $q_{i}$ for each channel. She also
generates the vector of public reconciliation strings $\mathbf{P}=\left(P_{1},\ldots,P_{n}\right)$
by calculating $P_{i}=q_{i}-\mu_{i}$ to aid Bob in his error correction
and to ensure matching secrets. He can then generate his quantization
level vector by calculating $q_{i}^{\prime}=\textrm{q}_{t}\left(\mu_{i}^{\prime}+P_{i}\right)$.
Both parties now have sufficient information to generate their candidate
secrets $secret$ and $secret^{\prime}$ by concatenating the resulting
binary strings.

\subsubsection{Key Verification Phase }

Finally, both parties proceed to verify if the secret keys are generated
successfully, i.e., if a mutual secret is established. After Bob has
finished his computations, he sends the hash value $\mathrm{h}\left(secret^{\prime}\right)$
of his secret string to Alice. Alice ensures successful key generation
by comparing Bob's value to her secret string. If the hash values
do not match, Alice can retry the key generation by increasing the
error count and choosing new tolerance values in the key generation
phase; redoing the sampling of the wireless channel is not necessary.
The approach used in our implementation uses a tolerance increase
of 0.5\,dB on each channel. However, our implementation on MICAz
sensor motes presented in the next section shows that with a tolerance
$t=1$, key agreement was reached in 94.6\,\% of the cases on the
first try. 

After finishing this step, both Alice and Bob share a secret key that
can be used to support security services.

\subsection{Protocol Optimizations\label{sub:Protocol-Optimizations}}

We experimented with some optimizations to increase the robustness
and secrecy of our protocol, and discuss some options in this section.
The later sections, however, base their analysis on the protocol described
in the previous section. 

The function \emph{chooseTolerance} can be improved further when the
tolerance values $t_{i}$ are chosen independently for each channel.
Our experiments show that only one or two channels have deviations
larger than the used tolerance values, and therefore prevent a successful
key generation. By choosing a higher tolerance value for single channels
only, Alice can start several key verification phases until the mismatching
channel is identified.

Our experiments show that the deviations can be approximated well
by a Normal distribution. This enables us to predict the success probability
of a protocol run, that can be used by Alice to aggressively choose
low tolerances in the beginning to increase the entropy of secret
strings, e.g., by initially achieving only a 56\,\% chance of a successful
key agreement with a tolerance value $t=0.4$.

\section{\label{sec:Experimental-Analysis}Implementation Results}

After the definition of the key generation protocol, the next interesting
aspect is how this protocol performs in real-world environments, and
how large the achievable secrecy and robustness is given realistic
propagation properties. With several experiments, these properties
are explored in detail in this section. We also show that the concept
is applicable on resource-constrained devices under realistic properties
of the wireless channel. The first part is focused on the robustness
and performance of the protocol, and in the second part the secrecy
is quantified empirically using the notion of information entropy.
These insights are also used as a basis and justification for the
analytical model, developed in Section~\ref{sec:Analysis}.

\begin{figure*}
\hspace*{\fill}\subfloat[Distribution of errors for the LOS experiment \label{fig:Distribution-of-errors-LOS}]{\hspace*{\fill}\psfragLabel{Differences between mean values [dB]}{Differences between mean values (dB)}
\psfragLabel{Density}{Density}

\psfragFixDiff{}
\psfragFixDensity{}\includegraphics[width=0.31\textwidth]{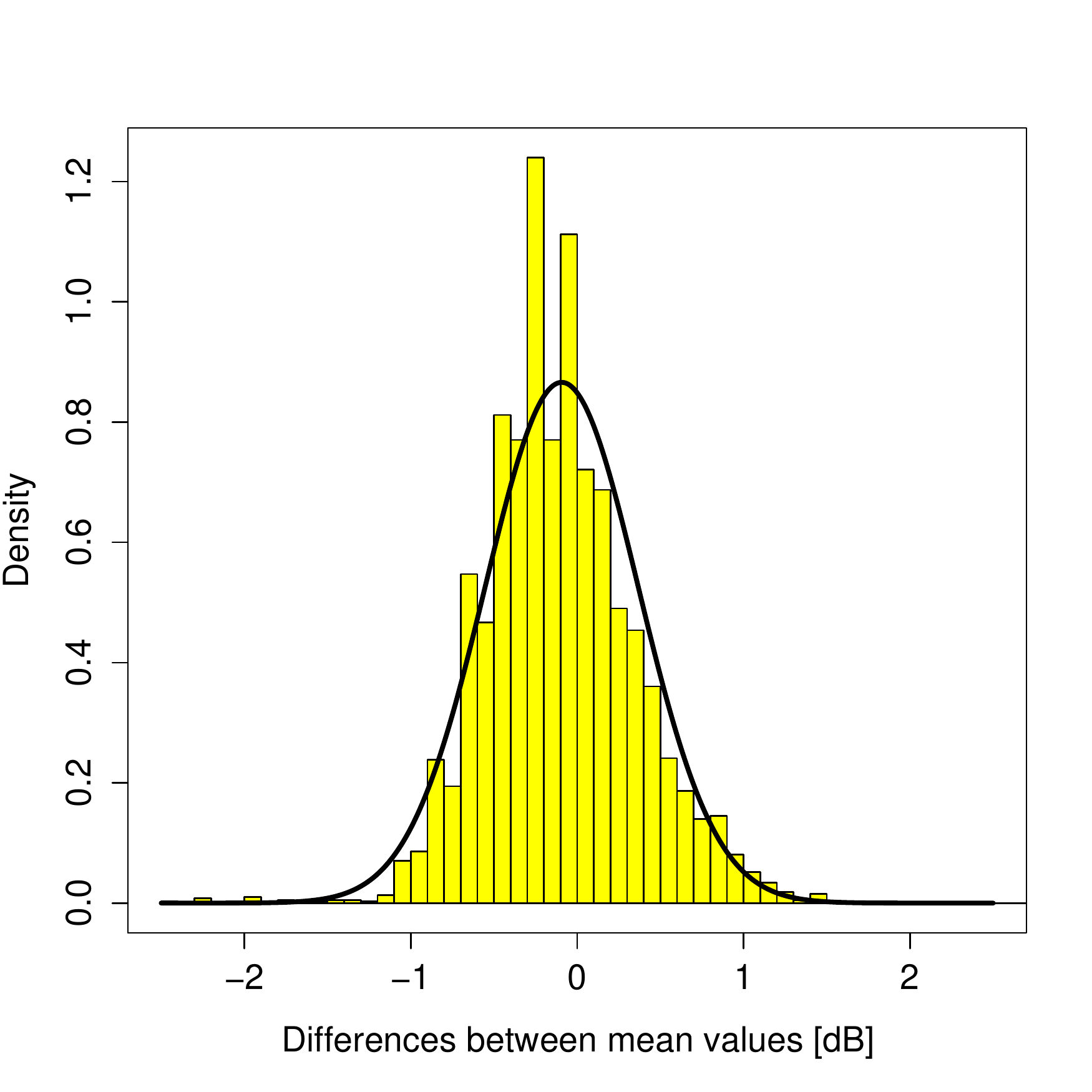}\hspace*{\fill}

}\hspace*{\fill}\subfloat[Distribution of errors for the non-LOS experiment \label{fig:Distribution-of-errors-nLOS}]{\hspace*{\fill}\psfragLabel{Differences between mean values [dB]}{Differences between mean values (dB)}
\psfragLabel{Density}{Density}

\psfragFixDiff{}
\psfragFixDensity{}\includegraphics[width=0.31\textwidth]{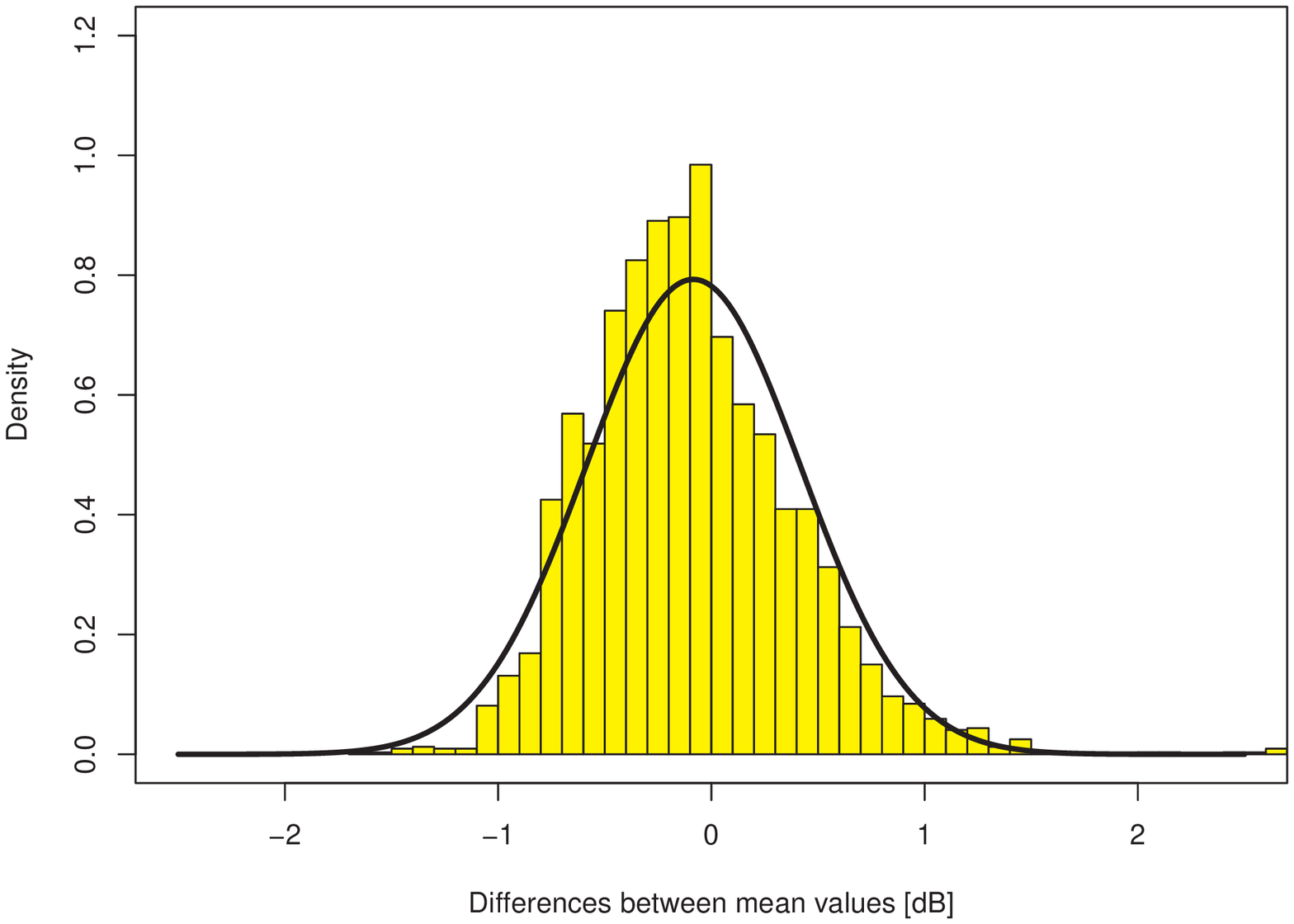}\hspace*{\fill}

}\hspace*{\fill}\subfloat[Success probability of key generation based on all positions. \label{fig:ECDF}]{\hspace*{\fill}\psfragLabel{Tolerance}{Tolerance}
\psfragLabel{Probability of successful key agreement}{Success Probability}

\psfragFixDiff{}
\psfragFixDensity{}\includegraphics[width=0.31\textwidth]{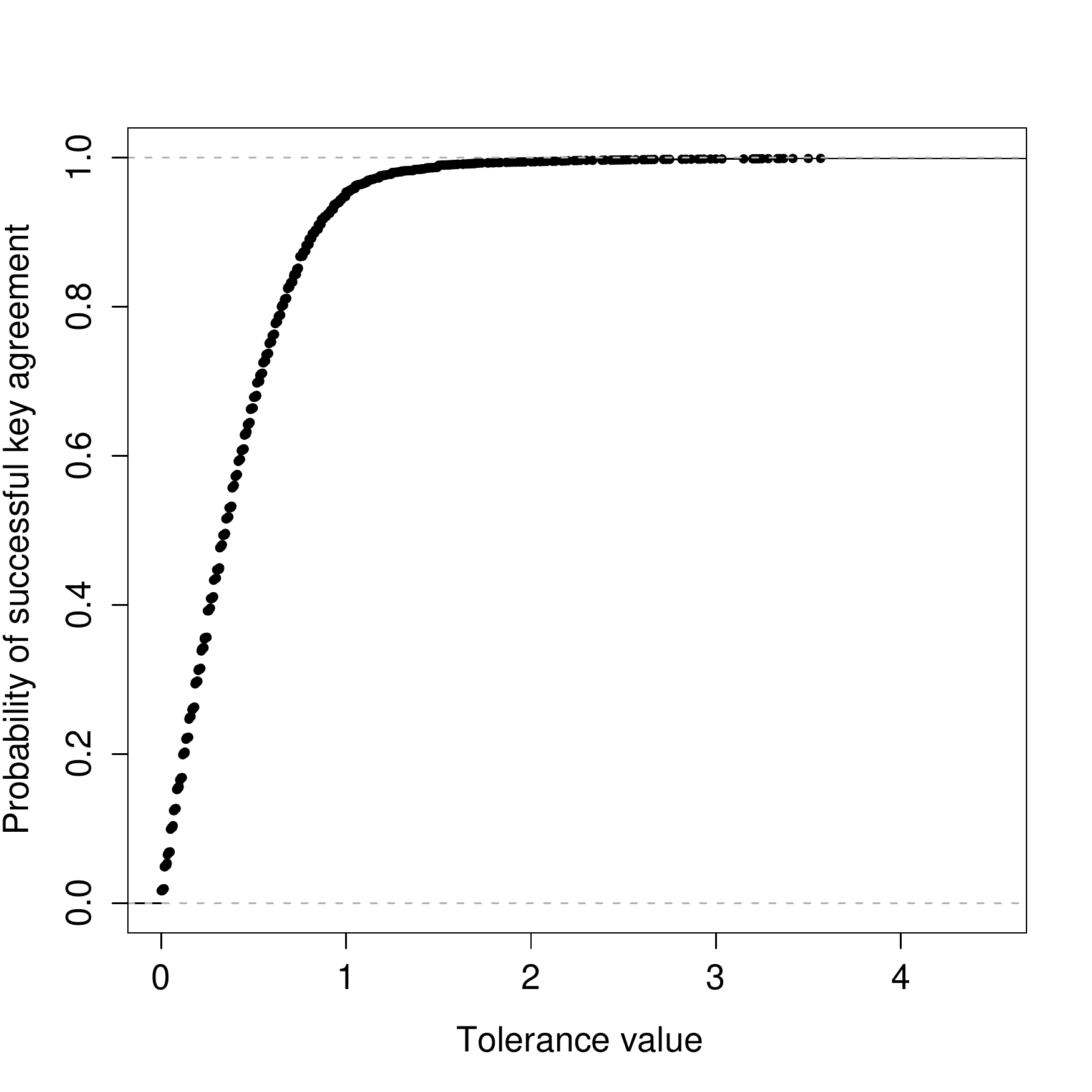}\hspace*{\fill}

}\hspace*{\fill}

\caption{Deviations in the channel and the resulting success rate of key generations
in our experiments. \label{fig:Robustness} }

\end{figure*}

\subsection{WSN Testbed and Methodology}

The experiments were conducted over several~days on a university~floor,
that is, an indoor setting across several rooms. During the measurements,
several wireless LAN access points were concurrently operating in
the 2.4\,GHz~band; so, the experiments were performed in a real-world
environment with unpredictable factors. The environment contains concrete
walls, as well as office furniture made of different materials. Especially
metal objects such as shelves and cabinets with good reflection properties
regarding electromagnetic waves were present. Thus, this set of environment
can be considered to generate a rich multipath effects, while it also
represents a typical indoor scenario. An additional factor for this
changing environment was the movement of people in corridors or office
rooms. 

Several different scenarios were considered to evaluate the impact
of positioning on secrecy and robustness. A large meeting room was
used for experiments, where the sensor motes always maintained a line
of sight connection, and several smaller office rooms were used to
quantify the impact of shadowing objects and walls. For each of these
scenarios, 250 positions were considered, and the distance was kept
constantly at 2.5 meters to avoid the influence of path loss effects.
In long-term and dynamic scenarios, these rooms and the connecting
corridors were used, and 1000 additional positions were tested with
mixed distances and obstacles. We used $k=16$ samples on each channel,
collected on $n=16$ channels.

\subsection{Protocol Robustness \label{sub:Imp-Protocol-Robustness}}

In order to evaluate the robustness of the protocol, a total of 1600~positions
of the two parties was tested, and the measurements and deviations
between the two parties recorded for each of the 16~channels. 

From the deviations $N=N_{\mathrm{Alice}}-N_{\mathrm{Bob}}$ observed,
we can see that they are bounded. The histogram of deviations is given
in Fig.~\ref{fig:Distribution-of-errors-LOS} and \ref{fig:Distribution-of-errors-nLOS},
which also shows that these deviations are fitted well by a zero-mean
Normal distribution with a standard deviation of $\sigma=0.461$\,dB
in the LOS experiment and $\sigma=0.503$\,dB in the non-LOS experiment.
The empirical distributions have even lighter tails than the fitted
Normal distributions. We can use this knowledge to evaluate the success
probability as described in Section~\ref{sub:Protocol-Optimizations}
for protocol optimizations. Based on the experiments, we can conclude
that the reciprocity of the wireless channel is very strong.

The success ratio of the protocol can be directly controlled by the
tolerance values of the code used, as codes with larger tolerance
values are able to correct stronger deviations. With a tolerance of
1\,dB, 94.6\,\% of the key agreements are successful on the first
run. This value is increased to 99.2\,\% with a tolerance of 2\,dB.
The empirical cumulative distribution function (ECDF) of all experiments
is shown in Fig.~\ref{fig:ECDF}. The majority of deviations are
below 2\,dB, and only a small number of extreme outliers were measured.
As the chosen tolerance value also has an impact on the secrecy of
the resulting bit string, a careful trade-off between secrecy and
robustness must be found.

\subsection{Evaluation of the Channel Entropy \label{sub:Imp-Secrecy}}

We evaluated the frequency-selective channel fading effects in two
different environmental settings: \emph{(i)} connections with line
of sight only; and \emph{(ii)} connections with obstacles in the direct
connection, that is, non-LOS connections. The LOS experiment was intended
as the worst-case scenario because a strong LOS component may be able
to dominate the multipath fading behavior. Yet, our experiments show
that this is not the case, and both experiments yield roughly the
same entropy. In all experiments, several different tolerance values
were considered to show the impact of this parameter on the secrecy.

\begin{figure}
\begin{framed}%
\input{figures/t-string/t-string.tex}\end{framed}

\caption{\label{fig:Example-of-T-String}A part of the T-string used for estimating
the Shannon entropy of codewords generated by our key-generation protocol.
This approach is based on encoding the codewords as ASCII strings
and analyzing their minimal representation.}

\end{figure}

The secrecy analysis focuses on the distribution of signal strength
measurements, especially on the entropy that these distributions offer.
The evaluation of the entropy for single channels is straightforward:
we use the empirical distribution to calculate $\mathrm{H}\left(C_{i}\right)$
for each of the $n$~channels individually, using the relative frequencies
as the estimates of codeword probabilities. For example, this analysis
shows that there are 3.5 secret bits available from each channel for
a tolerance value of $t=1$; a value of $t=0.5$ results in an increase
to 4.38\,bit. 

The joint entropy under the assumption of independent channels is
the sum of the channels' entropy values. However, the independence
cannot be assumed as the channels are within the coherence bandwidth,
and using the conventional approach to estimate the Shannon entropy
of dependent channels using sampling is not effective, as this becomes
prohibitive in spaces with larger dimensions. For example, to show
a joint secrecy of 45\,bit, at least $2^{45}$ samples must be collected.
Additionally, the unknown dependency structure of the generated secret
strings makes such quantification harder. The reason is that the Shannon
entropy operates on the knowledge of the underlying joint distribution,
which is unknown in our case. While in the next section we derive
a stochastic model for such analysis, we are still interested in finding
out how much uncertainty is present in the experimental data without
any assumptions on the underlying codeword distribution, i.e., without
requiring any a priori knowledge. The idea we follow is based on construction
complexity described by the notion of T-complexity~\cite{T-Entropy}.
T-complexity quantifies the difficulty to decompose input strings
into codewords of T-codes, i.e., the complexity when trying to find
the minimal representation of the input string. Speidel~\emph{et~al.}~\cite{EstimateShannon}
show in their work that T-complexity is the fastest to converge to
the true value of the Shannon entropy, and provide an algorithm that
enables fast computations of entropy values. The tool \texttt{tcalc}~\cite{tcalc},
developed by the same group, was used to evaluate our results. As
this tool operates on byte strings, we had to convert the lists of
quantized values to arrays consisting of different ASCII characters
as input. These characters were concatenated to form a large string
that can be used as input to \texttt{tcalc}. A part of the T-string
used is given in Fig. \ref{fig:Example-of-T-String}.

As a result, using this method we were able to capture the dependencies
between channels in the empirical data without explicitly knowing
them. The results from this analysis are discussed in the next subsection
and in Section~\ref{sub:Modeling-Channel-Dependency} we use them
for the validation of the derived stochastic model.%
\begin{figure}
\input{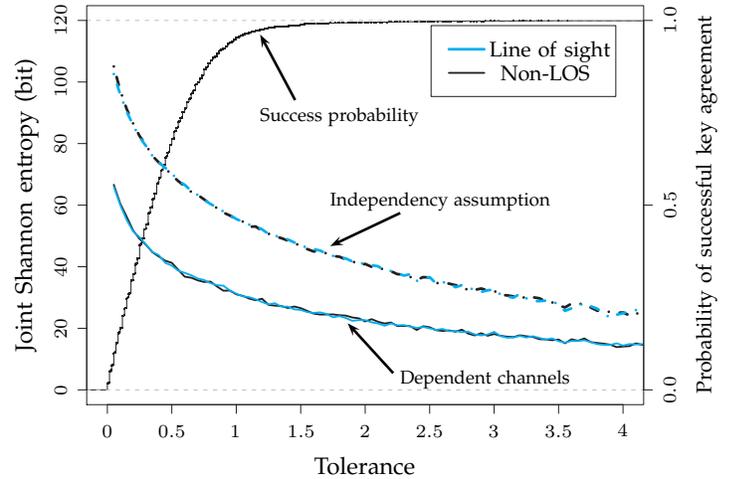}

\caption{Results for the implementation on MICAz sensor motes. The amount of
secrecy under different dependency assumptions is shown, with the
corresponding success probabilities of key agreement. \label{fig:Empirical-H-summary}}

\end{figure}

\subsubsection{Results from Experimental Analysis}

A comparison of results showing the available entropy from the experimental
data is shown in Fig.~\ref{fig:Empirical-H-summary}. With a tolerance
value of $t=1$, the entropy under independence assumption is 56\,bit
for both LOS and non-LOS connections. When considering the dependencies
in the measurements, 31\,bit of entropy can be achieved with the
limited number of channels and precision that the wireless sensor
mote hardware offers. Lower tolerance values can be used to increase
secrecy. For example, a tolerance value of 0.4, which results in a
56\,\% chance of successful key agreement, offers 45--50 secret bits
under dependent channels. 

The entropy of generated shared secrets in this settings can be compared
with conventional password-based security schemes and applied to the
protocols such as, for example, commitment-based authentication protocols
using short authenticated (e.g., \cite{ShortAuthenticatedStrings,ShortRandomPasswords,WeakPasswords,HumanPasswords}).
Similarly, protocols such as the Encrypted Key-Exchange (EKE) apply
short shared secrets for confidential exchange of public key material
(e.g., \cite{EncryptedKeyExchange,GeneEKE}). The shared secrets in
such applications are usually created by the user and contain approximately
18\,bit entropy due to dependency between characters (for a comprehensive
overview of password entropy, see \cite{EntropySymbolSeqs}). Since
these protocols protocols play an important role in wireless networks
as a part of device-pairing schemes, generating secrets from the wireless
channel can be seen as their valuable extension and alternative to
an user-required input of secrets.

\section{Increasing the Length of a Secret\label{sec:Analysis}}

The experimental analysis shows that the dependencies between channels
have considerable influence on the secrecy of the proposed protocol.
In contrast to previous section, we now develop a stochastic model
that makes these dependencies explicit and enables us to analyze and
predict ways to increase the achievable secrecy. Especially, we want
to answer questions such as: what is the impact of increasing the
number of available channels, and increasing the spacing between center
frequencies. To derive a realistic model of dependent wireless channels,
we start with fitting and validating the distribution of single channel
measurements and then extending it to a multivariate case, which captures
the dependencies between wireless channels. The model is validated
by comparing the resulting entropy values with our empirical results.

\subsection{Modeling Channel Dependency \label{sub:Modeling-Channel-Dependency}}

Frequently used distributions for large-scale models of wireless channels
are Rayleigh, Ricean, or Log-Normal \cite{Rappaport} depending on
the properties of the respective propagation environment. Also, in
scenarios common to WLANs and WSNs, where distances between transceivers
are short, the empirical data can be approximated by the Normal distribution
\cite{shengSpooing,Kaemarungsi,BoseNormRSS}. To find an adequate
distribution, we collected 4000 RSS sample means for each of the LOS
and non-LOS scenarios, where every RSS mean was calculated over 16
measurements, estimating the distribution parameters using Maximum
Likelihood Estimation (MLE). The resulting fit of the Rayleigh and
Normal distributions to the empirical data is shown in Fig.~\ref{fig:Normal-vs-Rayleigh}.
Additionally, we tested the normality of the sampled data using the
probability plot correlation coefficient test for normality (PPCC),
which is based on checking for linearity between the theoretical quantiles
and the sample data \cite{FillibenPPCC}. In fact, the goodness of
fit test confirms that the Normal distribution (correlation coefficient~$=0.992$)
can be assumed with an even higher confidence than the corresponding
Rayleigh distribution (correlation coefficient~$=0.967$). In this
case, the multivariate Normal distribution can be used to describe
the complex dependency structures of wireless channels by directly
estimating the covariance matrix from the empirical data. 

\begin{figure*}
\hspace*{\fill}\subfloat[Fitting of different distributions to the empirical data. \label{fig:Normal-vs-Rayleigh}]{\hspace*{\fill}\psfragLabel{Normalized RSS}{Normalized RSS}
\psfragLabel{Density}{Density}
\psfrag{Normal}[][]{\tiny{Normal}}
\psfrag{Rayleigh}[][]{\tiny{\!Rayleigh}}

\psfragTick{0.00}
\psfragTick{0.02}
\psfragTick{0.04}
\psfragTick{0.06}
\psfragTick{0.08}
\psfragTick{0.10}
\psfragTick{0.12}

\psfragFixTens[\tiny]{}\includegraphics[width=0.31\textwidth]{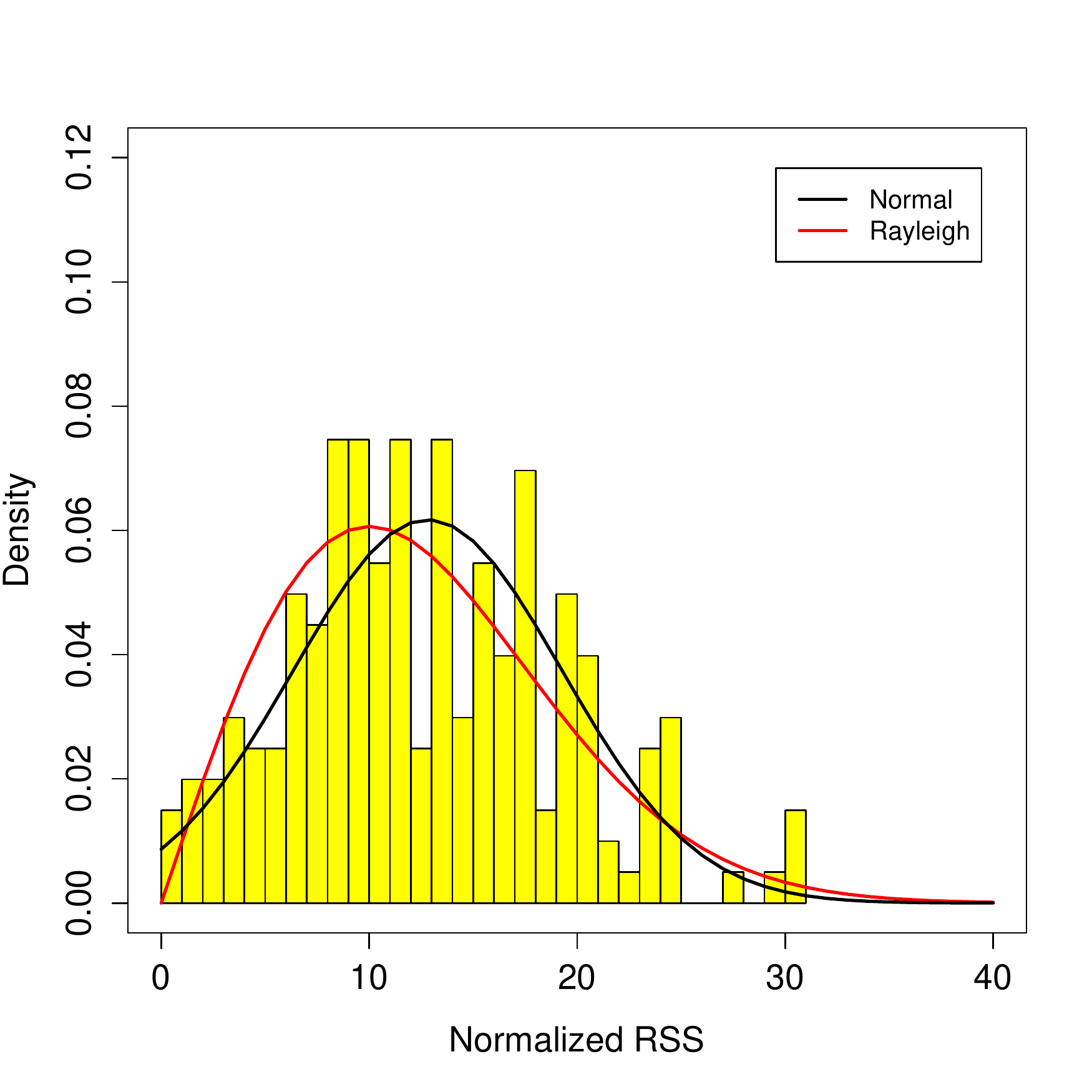}\hspace*{\fill}

}\hspace*{\fill}\subfloat[Rayleigh Probability-Probability test \label{fig:Test-Rayleigh}]{\hspace*{\fill}\psfragLabel{Rayleigh Theoretical Quantiles}{Rayleigh Theoretical Quantiles}
\psfragLabel{RSS Sample Data}{RSS sample data (dBm)}

\psfrag{Correlation Coefficient= 0.9674}[][]{\tiny{Corr. Coefficient=$0.9674$}}

\psfragFixRSSLevels{}
\psfragTick{0}
\psfragTick{5}
\psfragTick{10}
\psfragTick{15}
\psfragTick{20}
\psfragTick{25}
\psfragTick{30}
\psfragTick{35}\includegraphics[width=0.31\textwidth]{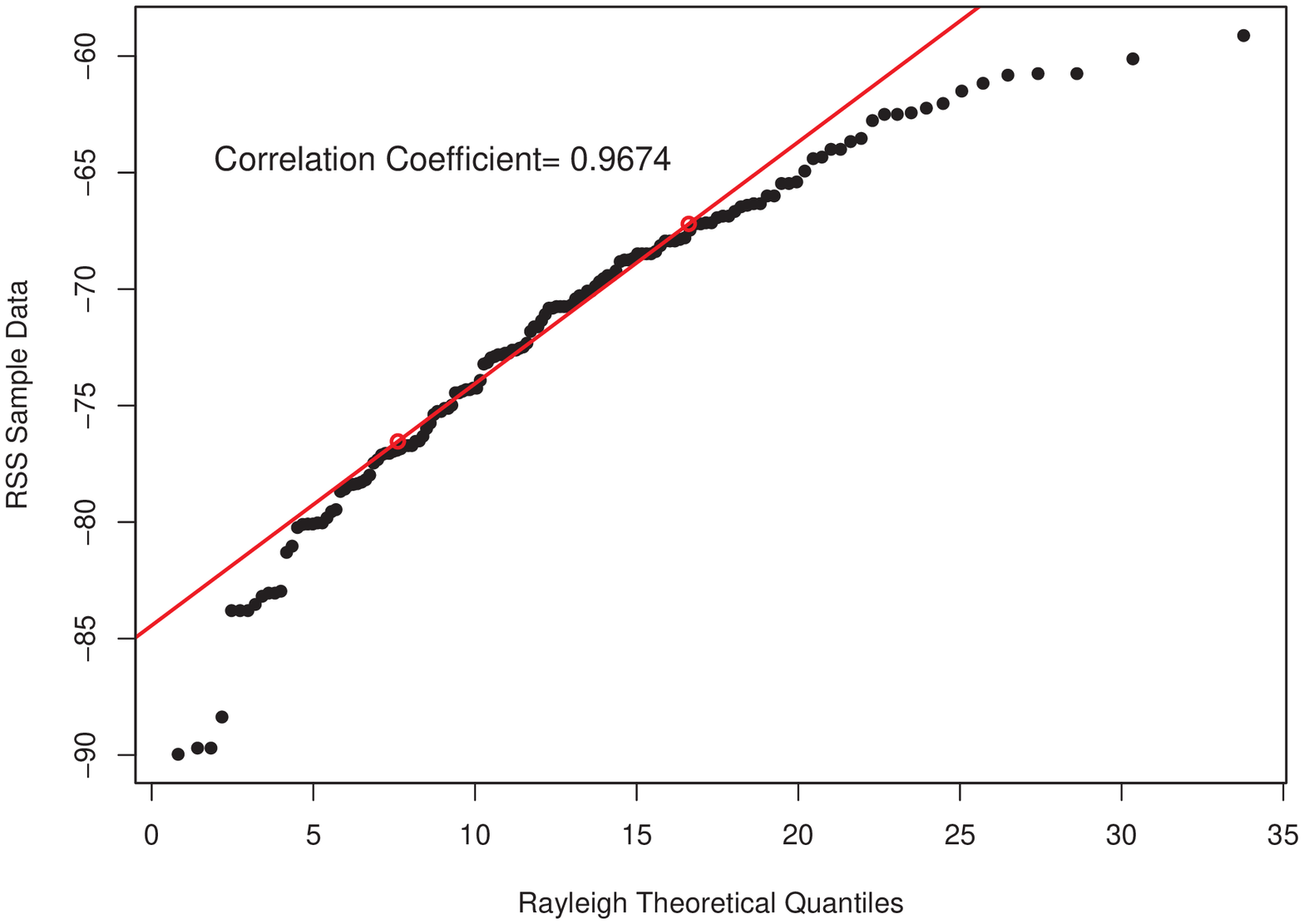}\hspace*{\fill}

}\hspace*{\fill}\subfloat[Normal Probability-Probability test \label{fig:Test-Normal}]{\hspace*{\fill}\psfragLabel{Normal Theoretical Quantiles}{Normal Theoretical Quantiles}
\psfragLabel{RSS Sample Data}{RSS sample data (dBm)}

\psfrag{Correlation Coefficient= 0.9927}[][]{\tiny{Corr. Coefficient=$0.9927$}}

\psfragFixRSSLevels{}
\psfragFixDiff{}\includegraphics[width=0.31\textwidth]{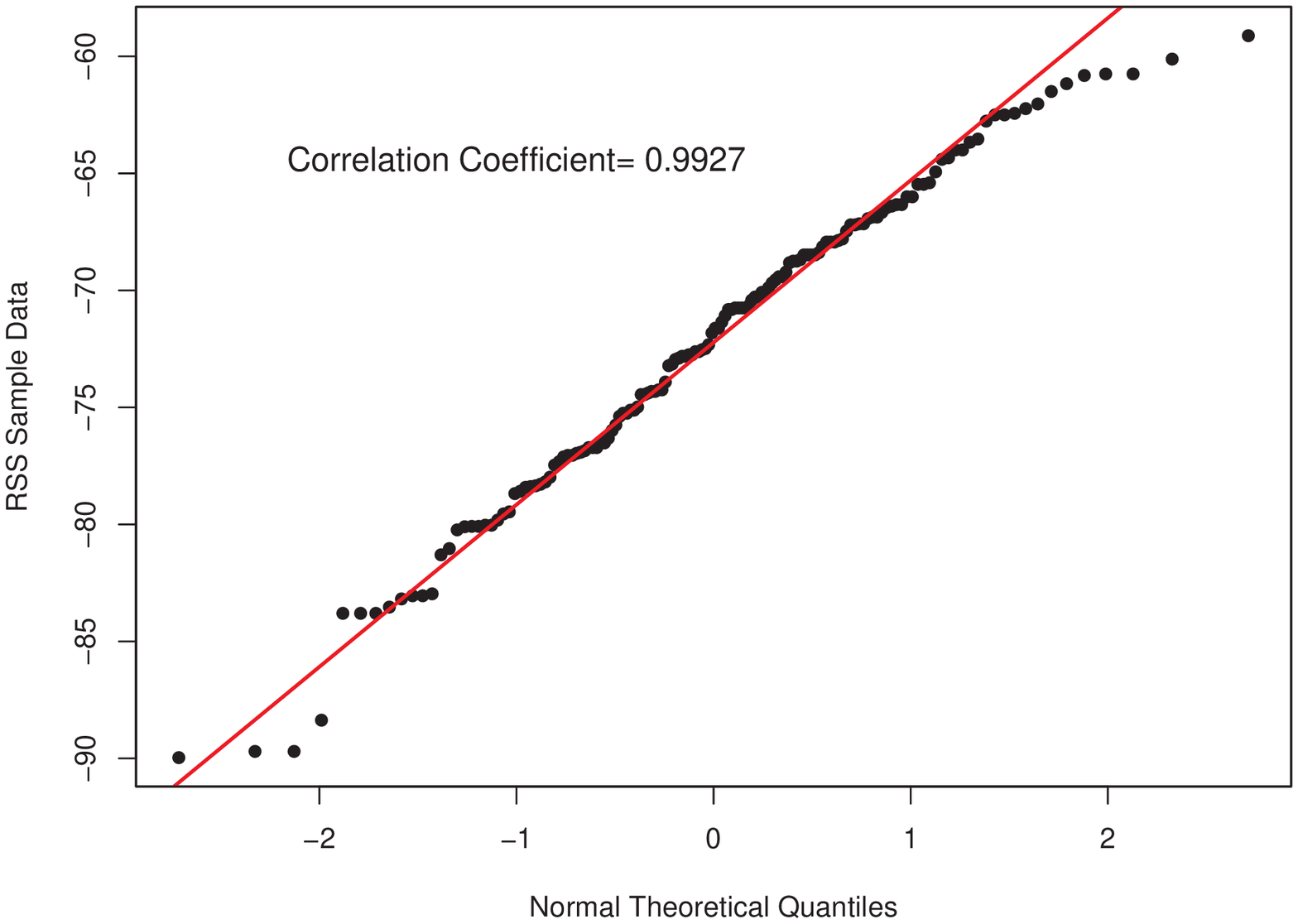}\hspace*{\fill}

}\hspace*{\fill}

\caption{Test for different distributions of the empirical data.\label{fig:Ivan's-figures}}

\end{figure*}

Hence, to analyze the dependencies of the joint distribution over
all 16 wireless channels, especially with respect to the joint entropy,
we model the signal strength values of different channels using a
single 16-dimensional multivariate Normal distribution. The distribution
parameter estimation is straightforward: the vector of mean values
$\mu$, which is in case of the Normal distribution already the MLE
for the population mean, and for the covariance matrix $\varSigma$
we used the MLE method: \[
\hat{\varSigma}=\frac{1}{k-1}\sum_{j=1}^{k}\left(\boldsymbol{m}^{(j)}-\boldsymbol{\mu}\right)\left(\boldsymbol{m}^{(j)}-\boldsymbol{\mu}\right)^{T}.\]

Finally, we validated the multivariate channel dependency model against
our empirical data by using the same error correction mechanism (described
in Section~\ref{sec:Protocol-Design}) to generate secret strings
and to compare the Shannon entropy of the empirical data with the
results of the model. The results of this evaluation are given in
Fig.~\ref{fig:Empirical-vs-Model}, which shows the resulting entropy
values for the non-LOS data applying the same analysis methods used
in the experimental analysis. The LOS experiment is omitted as the
behavior is similar. The model captures the dependency structure well,
resulting in a similar progression of the curve for the existing tolerance
values, although the entropy is slightly overestimated by the model.

\begin{figure}
\input{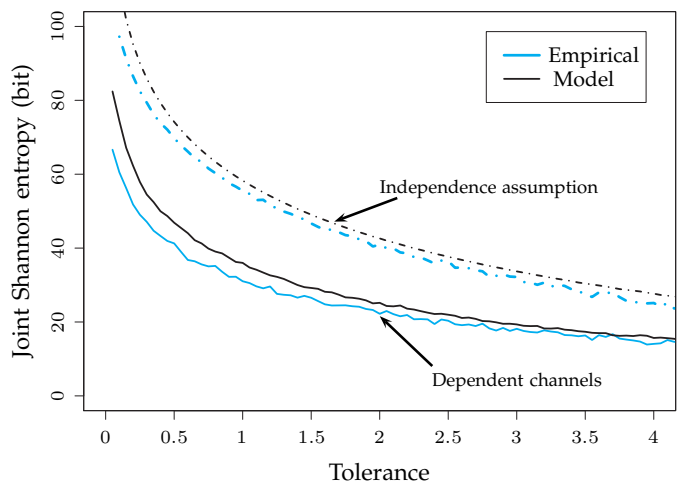}

\caption{Comparison of discrete entropy values based on RSS values generated
using the stochastic model. \label{fig:Empirical-vs-Model}}

\end{figure}

Using this model, we can estimate the amount of entropy if additional
resources are available, such as a higher number of channels or a
larger spacing between channels. We only need to consider the properties
of the covariance matrix~$\varSigma$ with respect to entropy. The
differential entropy (in natural units) of the multivariate Normal
distribution is given by \begin{equation}
H_{\mathrm{mv}\mathcal{N}}=\frac{1}{2}\ln\left(\left(2\pi\mathrm{e}\right)^{n}\det\varSigma\right),\label{eq:H_mvn}\end{equation}
depending on the number of channels~$n$ and the determinant of $\varSigma$.
The first-order effect of increasing the number of channels is easy
to quantify, the differential entropy is increased by $2.05$\,bit
for each additional channel. However, the relationship is not obvious
with respect to the determinant. In the case of independence, only
the main diagonal of the covariance matrix is populated, but in the
general case the complete matrix has an influence that is hard to
quantify.

\subsection{More Channels or Larger Frequency Spacing}

First, we consider the effects of the determinant on the security
given a larger number of channels. To this end, we extrapolate the
covariance matrix and evaluate the effect on the determinant. 

Two different prediction methods are used, one that extrapolates $\varSigma$
directly and another that also simulates the effect of larger spacing
between center frequencies and then extrapolates the matrix.

We used $\left(i{\times}i\right)$ sub-matrices with $i=1,\ldots,15$
of the matrix $\varSigma$ to predict the $16{\times}16$ matrix $\varSigma$.
Only the values contained in the sub-matrix are used, in the following
manner: each diagonal is treated independently, as it represents a
different lag in the covariances. The missing elements of the matrix
are chosen uniformly from a range between minimum and maximum values
on the respective diagonal. The results of this $16{\times}16$ prediction
for the non-LOS experiment are shown in Fig.~\ref{fig:16x16-Prediction}.
A sample of 100 extrapolated covariance matrices was used to predict
the known amount of differential entropy for 16 channels, the used
confidence level in the graph is 95\,\%. The horizontal line represents
a differential entropy using the correct $\varSigma$ from the experiments.
The predicted entropy values using different sub-matrix sizes are
shown, obtained from mean values of different uniform extrapolations.
Even with small $2{\times}2$ prediction matrices, it is possible
to estimate the entropy accurately. The evaluation for the LOS experiment
is not shown, but gave similar results. Thus, we can use the estimation
of $\varSigma$ to predict the secrecy from a larger number of channels.

The second matrix extrapolation method was used to evaluate the effects
of a larger spacing between the channels. Only every second (third,
$n$-th) diagonal was used and the remaining ones were removed for
this analysis. This simulates a channel spacing of 10\,MHz (15\,MHz,
$5n$\,MHz). This smaller matrix is then extrapolated in the same
fashion as described before. The quality of prediction is comparable
to the previous results. 

\begin{figure}
\input{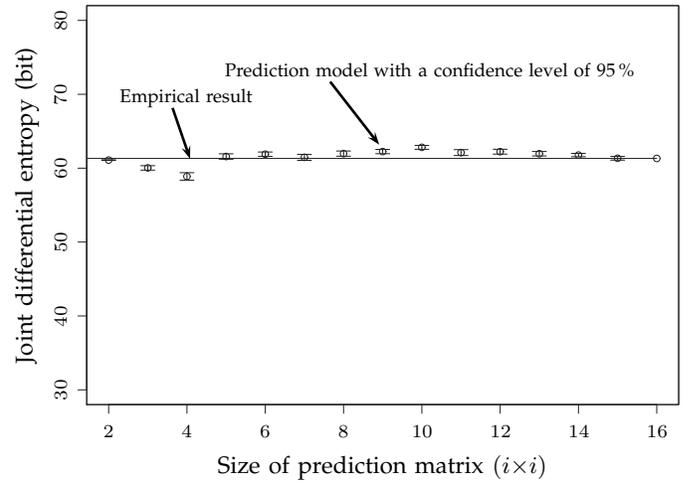}

\caption{Prediction of the differential entropy using only a subset of available
channels. Even with a small number of channels, an accurate prediction
is possible. \label{fig:16x16-Prediction}}

\end{figure}

Fig.~\ref{fig:Extrapolation-of-covariance} shows the increases of
entropy we can observe from our model. The figure shows the results
of the non-LOS experiment only, but the LOS experiment gave similar
results. The results are given in differential entropy, which does
not take the tolerances into account. The lowest line describes the
increase in joint differential entropy if we use the same determinant
we obtained from 16 channels. This results in an increase of 2.05\,bit
for each channel, but it is also a very conservative prediction, it
overestimates the dependencies between channels with center frequencies
far apart from each other. Using extrapolation based on the $16{\times}16$
matrix and calculating the entropy using Eq.~(\ref{eq:H_mvn}) and
the new $\varSigma$, we see an increase of 4.02\,bit for each additional
channel. The slashed line shows an additional gain if the channels
are spaced 10\,MHz\,apart, instead of the 5\,MHz spacing in our
experiments, yielding a 4.25\,bit increase. Our model shows that
there are several ways to increase the secrecy of the proposed protocol.
With measurements of higher precision it is possible to generate more
bits on each channel, but as this increases the hardware costs, it
is advisable to rather use a larger number of channels. 

\begin{figure}[h]
\psfragLabel[\small]{Number of channels}{Number of channels}%
\psfragLabel[\small]{Joint differential entropy}{Joint differential entropy (bit)}%
\psfragLegend{Fixed determinant}{\!Fixed $\det\left(\varSigma\right)$}%
\psfragLegend{Extrapolated determinant}{Extrapolated $\det\left(\varSigma\right)$}%
\psfragLegend{10MHz spacing}{~~~~~10\,MHz spacing}%
\psfragFixTens[\scriptsize]{}
\psfragTick[\scriptsize]{50}%
\psfragTick[\scriptsize]{100}%
\psfragTick[\scriptsize]{150}%
\psfragTick[\scriptsize]{200}%
\psfragTick[\scriptsize]{250}%
\psfragTick[\scriptsize]{300}%
\psfragFixDiff[\scriptsize]{}\includegraphics[width=1\columnwidth]{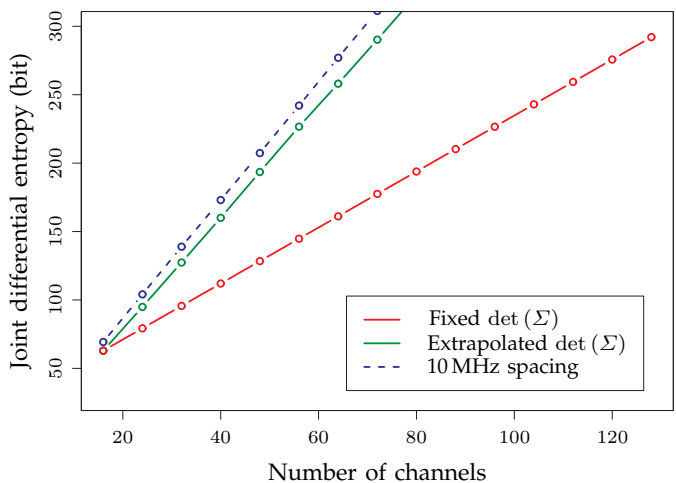}\caption{Extrapolation of covariance matrix $\varSigma$ for a larger number
of channels to evaluate of the model with respect to secrecy gains.
 \label{fig:Extrapolation-of-covariance}}

\end{figure}

\section{Conclusion \label{sec:Conclusion}}

Secret key generation and distribution poses one of the main security
challenges in wireless networks, especially in computation-limited
WSNs. In conventional security schemes, the wireless channel is usually
considered as a part of an adversarial toolbox which additionally
helps to launch different attacks by abusing its broadcast nature.
Yet, in recent years a number of papers following an alternative approach
to wireless security have demonstrated that the unpredictable and
erratic nature of wireless communication can be used to enhance and
augment conventional security designs. Taking advantage of physical
properties of signal propagation, mutual secrets between wireless
transmitters can be derived. While this approach for securing wireless
networks has been recently addressed in \cite{Telepathy,Envelopes},
both contributions require movement as the main generator of secret
material. Although valuable to mobile networks, such solutions are
not applicable to the majority of WSN applications which are based
on static sensor motes. 

The main focus of this work was to overcome this limitation. We started
by introducing a system model based on real-world measurements using
IEEE~802.15.4 technology, and describing building blocks of a novel
key generation protocol. To demonstrate its applicability, the protocol
was implemented and evaluated using MICAz sensor motes. Experiments
show that the protocol is able to successfully generate keys in over
95\,\% of the cases, irrespective of environmental properties. By
using only a very limited number of wireless channels, the proposed
protocol can already provide secrets up to 50\,bit, depending on
the wireless channel behavior. A stochastic model derived in this
work validated our experimental data and provided guidelines on how
to increase the length of the secret keys based on either increasing
the number of wireless channels or increasing the channel spacing.
For example, if the number of channels of the present IEEE~802.15.4
is set to 40, this protocol can generate up to 160\,bit secret keys
in static scenarios. 

The possibility to increase the length of a secret by using additional
wireless channels or larger frequency spacing is an interesting alternative
to computational-based approaches not only from the security perspective
but also from the network throughput perspective. For example, cognitive
radio is focused on increasing the utilization of limited radio resources
by dynamically adjusting the transmission to interference-free frequencies.
The key generation protocol introduced in this work can inherently
take advantage of such technologies.

Finally, this approach to key generation is intended to extend and
support conventional security designs as it only needs a limited number
of messages exchanges to generate shared secrets even on the currently
available, off-the-shelf WSN devices.

\bibliographystyle{plain}
\bibliography{references}

\end{document}